\documentclass{IEEEtran}
\usepackage{cite}
\usepackage{amsmath,amssymb,amsfonts}
\usepackage{graphicx}
\usepackage{textcomp}
\usepackage{color}
\usepackage[colorlinks=true,
linkcolor=black,
anchorcolor=black,
citecolor=black
]{hyperref}
\usepackage{epstopdf}
\usepackage{bm}
\usepackage{subcaption}
\usepackage{graphicx}
\usepackage{subcaption}
\usepackage[linesnumbered,ruled,vlined]{algorithm2e}
\usepackage{tikz}
\usepackage{pgfplots}
\pgfplotsset{compat=newest}
\usepackage{dsfont}
\usepackage{soul}

\newenvironment{proof}{{\it Proof:}}{\hfill $\blacksquare$}

\newcounter{theorem}
\newtheorem{theorem}{\it Theorem}
\newcounter{lemma}
\newtheorem{lemma}{\it Lemma}
\newcounter{corollary}
\newtheorem{corollary}{\it Corollary}
\newcounter{proposition}
\newtheorem{proposition}{\it Proposition}
\newcounter{remark}
\newtheorem{remark}{\it Remark}
\newcounter{definition}
\newtheorem{definition}{\it Definition}
\newcounter{assumption}
\newtheorem{assumption}{\it Assumption}
\newcounter{example}
\newtheorem{example}{\it Example}
\def\BibTeX{{\rm B\kern-.05em{\sc i\kern-.025em b}\kern-.08em
    T\kern-.1667em\lower.7ex\hbox{E}\kern-.125emX}}

\begin{document}

\title{Signed DeGroot--Friedkin Dynamics with Interdependent Topics}
\author{Yangyang Luan, Muhammad Ahsan Razaq, Xiaoqun Wu, and Claudio Altafini
	\thanks{This work was supported in part by the National Natural Science Foundation of China under Grant 62573301, by the Startup Grant from Shenzhen University, by the Swedish Research Council under Grant 2020-03701 and Grant 2024-04772, and by the ELLIIT framework program at Link\"{o}ping University. \textit{(Corresponding authors: Claudio Altafini; Xiaoqun Wu.)}}
	\thanks{Yangyang Luan is with the School of Mathematics and Statistics, Wuhan University, Wuhan 430072, China, and also with the Division of Automatic Control, Department of Electrical Engineering, Link\"{o}ping University, SE-58183 Link\"{o}ping, Sweden (e-mail: luanyy\_1704@whu.edu.cn).}
	\thanks{Muhammad Ahsan Razaq and Claudio Altafini are with the Division of Automatic Control, Department of Electrical Engineering, Link\"{o}ping University, SE-58183 Link\"{o}ping, Sweden (e-mail: muhammad.ahsan.razaq@liu.se; claudio.altafini@liu.se).}
	\thanks{Xiaoqun Wu is with the College of Computer Science and Software Engineering, Shenzhen University, Shenzhen 518060, China (e-mail: xqwu@szu.edu.cn).}
}

\maketitle

\begin{abstract}
This paper investigates DeGroot--Friedkin (DF) dynamics over signed influence networks with interdependent topics.
We propose a multi-topic signed framework that combines repelling interpersonal interactions with cross-issue self-appraisal, examining how antagonism and topic interdependence shape the evolution of agent-level social power.
When the logic matrices (for topic interdependence) of all agents share a common dominant left eigenvector, we identify structural conditions under which the original dynamics admit an exact reduction to an explicit scalar DF map.
This yields a complete classification of limiting social power configurations into pluralistic, mixed, and vertex-dominant types.
In all three cases, the dynamics are globally convergent, and in the first two the ordering induced by the interaction centrality is preserved.
We further show local robustness under small heterogeneous perturbations of the logic matrices.
We also clarify what changes when this common-eigenvector structure is lost.
These results extend signed social power dynamics beyond the standard nonnegative scalar setting and shed light on the robustness and scope of centrality-based social power formation in multi-topic signed influence systems.
\end{abstract}


\section{Introduction}
\label{sec:Intro}

In recent years, network science and control theory have become important tools for modeling opinion dynamics, providing mathematical frameworks for understanding how beliefs spread, interact, and stabilize within social groups~\cite{proskurnikov2017tutorial}.
A cornerstone in this area is the DeGroot model~\cite{degroot1974reaching}, in which each individual repeatedly updates her opinion by averaging those of her neighbors. Although elegant and analytically tractable, the model treats interpersonal influence as fixed.
In many social settings, however, this influence is itself endogenous: agents who are more successful in shaping collective outcomes tend to gain confidence and exert greater influence in subsequent issues.
Capturing this feedback between opinion formation and evolving influence is one of the central motivations for studying social power dynamics.

The DeGroot--Friedkin (DF) model addresses this point by coupling DeGroot opinion formation with Friedkin's reflected-appraisal mechanism~\cite{friedkin2011formal,jia2015opinion}. In this framework, self-weights are not prescribed a priori, but emerge endogenously through repeated issue discussions. An agent who more strongly affects the outcome of one issue assigns herself a larger self-weight for the next one, generating a coevolutionary process between interpersonal influence and self-appraisal.
Since its introduction, the DF model has stimulated a broad literature on reducible networks~\cite{jia2017opinion}, switching interaction structures~\cite{chen2018social,ye2018evolution}, alternative updating rules~\cite{xu2015modified,jia2019opinion}, and heterogeneous psychological traits such as stubbornness~\cite{tian2021social}.
Yet most available results remain confined to cooperative and scalar settings, where interpersonal influence is nonnegative and discussion topics are treated as independent.

A major limitation of this paradigm is the absence of antagonism. Real influence networks are often signed, containing both trust and opposition.
Existing signed extensions of DF dynamics largely adopt the ``opposing'' framework~\cite{bhowmick2020evolution,xue2020evolution,shi2019dynamics}, which preserves analytical tractability through a particular normalization of the signed Laplacian: its diagonal entries are defined by the sums of the absolute values of the corresponding row entries.
This construction is analytically convenient, but it also narrows the range of admissible asymptotic behaviors. Indeed, under standard connectivity assumptions, the resulting dynamics exhibit a rather restrictive dichotomy: structural balance yields bipartite consensus, whereas structural unbalance drives all opinions to zero.

A more general alternative is the ``repelling'' framework~\cite{shi2019dynamics,fontan2022multiagent}, in which the signed Laplacian contains on the diagonal the row sums of the signed entries, thereby preserving the original signed algebraic structure more directly.
The repelling framework can accommodate a wider range of signed behaviors than the opposing one, with the structurally balanced case appearing only as a special subregime.
This added generality, however, comes at an analytical price: the repelling Laplacian always has zero as an eigenvalue, but the dynamics may be unstable, and even when stable, convergence to consensus cannot be guaranteed~\cite{razaq2025signed,razaq2024multidimensional}. Consequently, extending reflected appraisal to repelling signed interactions is not a straightforward extension of the standard nonnegative DF theory.
This motivates an analysis based on eventual stochasticity (ES) and additional contractivity requirements, such as the single-leader contractivity (SLC) condition introduced in this paper, to ensure that the resulting self-appraisal map remains meaningful and dynamically well behaved.

Another simplification in the existing DF literature is the scalar treatment of issues.
In many realistic decision-making scenarios, the agents reason simultaneously over multiple logically interdependent topics~\cite{parsegov2016novel,ye2020continuous,luan2025coevolutionary}. Such interdependence is not a minor modeling detail, nor can it be reduced to stubbornness or inertia. Instead, it governs how socially received opinions are internally processed before being reflected in the next self-appraisal step. Although logical interdependence has been studied extensively in multidimensional DeGroot-type models, its effect on reflected-appraisal dynamics over signed influence networks remains largely unexplored.
Accordingly, once one moves beyond scalar cooperative discussions, a natural question arises: is the long-term agent-level social power profile determined solely by the signed interaction topology, or can it also depend on the internal topic-coupling structure encoded by the agents' logic matrices?

This paper develops a signed multi-topic DF framework that addresses these questions.
We consider a model in which each issue consists of a multidimensional DeGroot stage with repelling interpersonal interactions, followed by a Friedkin self-appraisal update across issues.
In the shared topic-weight regime, where the logic matrices of the agents all share a common dominant left eigenvector, we show that the problem admits an exact algebraic reduction.
More precisely, under the proposed SLC condition on the interaction component, together with a coupled spectral regularity condition, the original self-appraisal dynamics reduce exactly to an explicit scalar DF map on the simplex.
This reduction yields a complete classification of equilibria and global behavior: depending on the interaction centrality profile, the limiting social power configuration is either a unique interior equilibrium or a vertex of the simplex.
In the interior case, the limiting configuration is either ``pluralistic'' (mixed but with no agent having the majority of social power), or ``majoritarian'' (mixed but with a majority leader).
Moreover, whenever an interior equilibrium exists, it preserves the ordering induced by the interaction centrality vector.

We then move beyond exact separability in two directions.
First, for sufficiently small heterogeneous perturbations of the logic matrices, we quantify how the eigenvalue $1$, its associated normalized left eigenvector, and the induced update map vary under perturbation, and use these estimates to show that locally attractive interior equilibria persist near the reference regime.
This shows that the centrality-based characterization obtained in the shared topic-weight regime exhibits local robustness to mild violations of exact separability.
Second, we clarify the role of the heterogeneous logic under the present self-appraisal aggregation rule.
Although a heterogeneous logic can generate substantial disparities in topic-level transfer before the self-appraisal aggregation, these pre-aggregation differences disappear once the current self-appraisal update is applied, and only the interaction component matters. 
Taken together, these results identify both a local robustness property near the shared topic-weight regime and a structural limitation of the current aggregation mechanism.

The remainder of this paper is outlined  as follows.
Section~\ref{sec:Preliminaries} introduces preliminaries on signed graphs and matrix spectral properties. 
Section~\ref{sec:ModelFormulation} formulates the signed multi-topic DF model.
Section~\ref{sec:baseline-sharedc} develops the exact reduction in the shared topic-weight regime and derives the equilibrium and convergence results.
Section~\ref{sec:beyond-separable} examines local robustness and the role of heterogeneous logic beyond exact separability.
Section~\ref{sec:Simulations} presents numerical simulations, and Section~\ref{sec:Conclusion} concludes the paper.
Some technical lemmas and proofs are in the Appendix.

A preliminary version of this work, to appear in the proceedings of the 23rd IFAC World Congress~\cite{razaq2026ifac_signed_df}, considers a scalar signed DF model and introduces the SLC condition in that setting.
The multi-topic framework developed here, the exact reduction in the shared topic-weight regime, the equilibrium classification with global convergence, and the perturbation and aggregation results are new.

\section{Preliminaries}
\label{sec:Preliminaries}

\subsection{Notation}
Scalars are denoted by lowercase letters (e.g., $x$), vectors by bold lowercase letters (e.g., $\mathbf{x}$, $\bm{x}$), and matrices by uppercase letters (e.g., $A$, $\cal{A}$).
The vector of all ones and the vector of all zeros in $\mathbb{R}^n$ are denoted by $\mathds{1}_n$ and $\mathbf{0}_n$, respectively; subscripts are omitted when dimensions are clear.
The vector $\mathbf{e}_i$ denotes the $i$-th canonical basis vector in $\mathbb{R}^n$.
For a vector $\mathbf{x} \in \mathbb{R}^n$, let $\|\mathbf{x}\|_1 := \sum_{i=1}^n |x_i|$ denote the $\ell_1$-norm, and let $\operatorname{diag}(\mathbf{x})$ denote the diagonal matrix with the entries of $\mathbf{x}$ on its main diagonal.
Conversely, for a matrix $A \in \mathbb{R}^{n \times n}$, $\operatorname{diag}(A) \in \mathbb{R}^{n}$ denotes the vector formed by its diagonal entries.
The block diagonal matrix formed by matrices $A_1, \dots, A_n$ is denoted by $\operatorname{blkdiag}(A_1, \dots, A_n)$.
Let $I_n$ be the $n\times n$ identity matrix, and $\otimes$ denote the Kronecker product.
The spectrum of $A \in \mathbb{R}^{n \times n}$ is denoted by $\Lambda(A) = \{ \lambda_i(A) \}_{i=1}^n$, and its spectral radius by $\rho(A) := \max_{\lambda \in \Lambda(A)} |\lambda|$.
A real eigenvalue $\lambda_\star\in\Lambda(A)$ is called dominant if $|\lambda_\star|>|\lambda|$ for every $\lambda\in\Lambda(A)\setminus\{\lambda_\star\}$; the associated left and right eigenvectors are called dominant eigenvectors.
Component-wise inequalities for vectors and matrices are denoted by $>,<,\ge,\le$. 
The induced matrix $\ell_1$-norm is $\|A\|_1 := \max_j \sum_i |a_{ij}|$.
Given a symmetric positive definite matrix $P\succ 0$, the induced $P$-norm is
\[
\|A\|_P := \sqrt{\lambda_{\max} \left(P^{-1/2}A^\top P A P^{-1/2}\right)} .
\]
The $n$-simplex is defined by $\Delta_n := \{ \mathbf{x} \in \mathbb{R}^n \mid \mathbf{x} \ge \mathbf{0}, \mathds{1}^\top \mathbf{x} = 1 \}$ and the $n$-hyperplane by $\mathcal{H}_n := \{ \mathbf{x} \in \mathbb{R}^n \mid \mathds{1}^\top \mathbf{x} = 1 \}$.
The interior and boundary of $\Delta_n$ are $\mathrm{int}\Delta_n := \{ \mathbf{x} \in \Delta_n \mid \mathbf{x} > \mathbf{0} \}$ and $\partial\Delta_n := \Delta_n \setminus \mathrm{int}\Delta_n$, respectively.

\subsection{Signed Graphs}

A signed directed graph (digraph) is a triplet $\mathcal G=(\mathcal V,\mathcal E,A)$, where $\mathcal V=\{1,\dots,n\}$ is the node set, $\mathcal E\subseteq \mathcal V\times\mathcal V$ is the edge set, and $A=[a_{ij}]\in\mathbb R^{n\times n}$ is the signed weighted adjacency matrix.
A directed edge $(j,i) \in \mathcal{E}$ exists if and only if $a_{ij} \neq 0$.
Unlike unsigned graphs where $a_{ij} \ge 0$, a signed graph allows $a_{ij}<0$, capturing antagonistic interactions.
For nonnegative matrices, $A$ is said to be row-stochastic if $A \ge 0$ and $A\mathds{1} = \mathds{1}$, and doubly stochastic if both $A$ and $A^\top$ are row-stochastic.

\subsection{Spectral Properties of Signed Matrices}

The classical Perron--Frobenius (PF) theory for nonnegative matrices ensures the existence of a positive eigenvector associated with the spectral radius. 
For signed matrices, analogous asymptotic properties are defined as follows~\cite{fontan2022multiagent,razaq2025signed}.

\begin{definition}[Perron--Frobenius (PF)]
	A matrix $A \in \mathbb{R}^{n\times n}$ satisfies the right (resp.\ left) PF property, denoted by $A\in \mathcal{PF}$ (resp.\ $ A^{\top} \in \mathcal{PF}$), if there exists a simple real and positive eigenvalue $\lambda_1\in\Lambda(A)$ such that $\lambda_1 > |\lambda_j(A)|~\forall\lambda_j\neq\lambda_1$, and the associated right (resp.\ left) eigenvector is strictly positive.
\end{definition}

\begin{definition}[Stochastic PF (SPF)]
	A matrix $A \in \mathbb{R}^{n \times n}$ is termed SPF if $A \in \mathcal{PF}$, $A\mathds{1} = \mathds{1}$, and $\rho(A) = 1$.
\end{definition}

\begin{definition}[Eventually Stochastic (ES)]
	A matrix $A \in \mathbb{R}^{n \times n}$ is termed ES if both $A \in \mathcal{PF}$ and $A^\top \in \mathcal{PF}$ hold, while satisfying $A\mathds{1} = \mathds{1}$ and $\rho(A) = 1$.
\end{definition}

\begin{remark}
	If $A$ is SPF, then $\lambda=1$ is the simple dominant eigenvalue with strictly positive right eigenvector $\mathds{1}$.
	If $A$ is further ES, then the normalized left Perron eigenvector is also uniquely defined and strictly positive, and therefore belongs to $\mathrm{int} \Delta_n$.
	These properties will be repeatedly used in the signed setting considered in this paper.
\end{remark}

\section{Model Formulation}
\label{sec:ModelFormulation}

In this section, we formulate a topic-coupled DF model that consists of two interrelated components:
(i) a multidimensional DeGroot stage for intra-issue opinion formation; (ii) a Friedkin self-appraisal stage for inter-issue self-weight updating.
This framework allows us to investigate how antagonistic interpersonal influence and heterogeneous belief structures collectively shape long-term self-weight allocations across issues.

\subsection{Multidimensional DeGroot Dynamics}

Consider a group of $ n $ agents discussing $m$ logically interdependent topics over a sequence of issues $\mathcal{S}= \{0,1,2,\ldots\} $.
Let $\mathcal I=\{1,\ldots,n\}$ denote the agent set.
For each issue $ s\in\mathcal{S} $, agent $ i\in\mathcal{I} $ holds opinions on the $m$ topics at time $t$, denoted by $\mathbf{y}_i(s,t) = [y_i^1(s,t),\ldots,y_i^m(s,t)]^{\top} \in \mathbb{R}^{m}$, where $y_i^p(s,t)$ is agent $i$'s opinion on topic $p$.

During discussions over issue $s$, each agent $i$ updates her opinions according to the linear rule:
\begin{equation}\label{equ:DynamicsOfEachAgent}
	\mathbf{y}_i(s,t+1) = \sum_{j=1}^{n} w_{ij}(s) C_i \mathbf{y}_{j}(s,t),
\end{equation}
where $w_{ij}(s)\in\mathbb{R}$ denotes the (possibly signed) interpersonal influence coefficient that agent $i$ assigns to agent $j$ on issue $s$.
We collect these coefficients into the signed influence matrix $W(s)=[w_{ij}(s)]\in\mathbb{R}^{n\times n}$.
Throughout the paper, we adopt the standard row-sum normalization, i.e., $W(s)\mathds{1}=\mathds{1}, \forall s\in\mathcal{S}$.

The logical interdependencies among the $m$ topics perceived by agent $i$ are encoded through  a logic matrix $ C_i=[c_{pq,i}]\in\mathbb{R}^{m\times m}$, where $c_{pq,i}$ quantifies the dependence of topic $p$ on topic $q$ within agent $i$'s belief structure~\cite{parsegov2016novel,friedkin2016network}.
The matrices $\{C_i\}_{i\in\mathcal{I}}$ are allowed to be heterogeneous across agents, i.e., there may exist $i\neq j$ such that $C_i\neq C_j$.

Our primary interest is the evolution of self-weights across issues, i.e., the diagonal entries $w_{ii}(s)$.
Define $ x_{i}(s) := w_{ii}(s) $ and $\mathbf{x}(s) := [x_1(s),\ldots,x_n(s)]^\top$.
Because $W(s)\mathds{1}=\mathds{1}$, one obtains $\sum_{j\neq i} w_{ij}(s)=1-x_i(s)$.
Following the classical DF model~\cite{jia2015opinion}, we assume that the relative off-diagonal interaction proportions are issue-independent.
Specifically, for $j\neq i$, one has $w_{ij}(s)=(1-x_i(s))r_{ij}$, where $r_{ij}$ is a constant relative interaction coefficient (possibly signed), with $r_{ii}=0$ for all $i\in\mathcal I$.
Accordingly,
\begin{equation}\label{equ:DecompositionW}
	W(s)
	=
	\operatorname{diag}(\mathbf{x}(s))
	+
	\bigl(I_n-\operatorname{diag}(\mathbf{x}(s))\bigr)R,
\end{equation}
where $R=[r_{ij}] \in \mathbb R^{n\times n}$ satisfies $\operatorname{diag}(R)=\bm 0$ and $R\mathds{1}=\mathds{1}$.

Let $ \mathbf{y}(s,t)=[\mathbf{y}_1(s,t)^{\top},\ldots,\mathbf{y}_n(s,t)^{\top}]^{\top} \in \mathbb{R}^{nm} $.
Define $ \mathcal{C}_{[n]} := \operatorname{blkdiag}(C_1,\dots,C_n) \in \mathbb{R}^{nm \times nm} $.
Aggregating \eqref{equ:DynamicsOfEachAgent} over all agents yields
\begin{equation}\label{equ:DynamicsOfAllAgents}
	\mathbf{y}(s,t+1) = \mathcal{C}_{[n]}\bigl(W(s)\otimes I_m\bigr)\mathbf{y}(s,t).
\end{equation}
Denote the associated system matrix by
\begin{equation}\label{equ:SystemMatrix}
	\mathcal{A}(\mathbf{x}(s)) = \mathcal{C}_{[n]}\bigl(W(\mathbf{x}(s))\otimes I_m\bigr).
\end{equation}
Hence, \eqref{equ:DynamicsOfAllAgents} defines a multidimensional DeGroot-type update in which opinion formation is jointly shaped by signed interpersonal influence and within-agent topic interdependence.

\subsection{Friedkin Self‑Appraisal Dynamics}
\label{subsec:SelfAppraisal}

To define the self-appraisal update in the signed setting, we work on the affine hyperplane $\mathcal{H}_n$.
For any $\mathbf{x}\in\mathcal{H}_n$ such that the system matrix $\mathcal{A}(\mathbf{x})$ has $1$ as a simple dominant eigenvalue, let $\bm{\mu}(\mathbf{x})\in\mathbb{R}^{nm}$ denote the normalized left eigenvector satisfying $\bm{\mu}(\mathbf{x})^{\top}\mathcal{A}(\mathbf{x})=\bm{\mu}(\mathbf{x})^{\top}$ and $\bm{\mu}(\mathbf{x})^{\top}\mathds{1}_{nm}=1$.
Sufficient structural conditions guaranteeing this property are established in later sections.

Partition $\bm\mu(\mathbf{x})=[\bm\mu_1(\mathbf{x})^\top,\ldots,\bm\mu_n(\mathbf{x})^\top]^\top$,
where $\bm{\mu}_i(\mathbf{x})\in\mathbb{R}^m$ represents the topic-level influence profile associated with agent $i$.
In signed settings, entries of $\bm{\mu}(\mathbf{x})$ may take negative values due to antagonistic interactions.
We adopt the following aggregation rule to obtain the agent-level update:
\begin{equation}\label{equ:AggregatedPower}
	\phi_i(\mathbf{x}):=\mathds{1}_m^\top \bm\mu_i(\mathbf{x})
	=\sum_{p=1}^m \mu_{i,p}(\mathbf{x}).
\end{equation}
Let $\bm\phi(\mathbf{x})=[\phi_1(\mathbf{x}),\ldots,\phi_n(\mathbf{x})]^\top$.
Given that $\sum_{i=1}^n \phi_i(\mathbf{x})
= \sum_{i=1}^n \mathds{1}_m^\top \bm\mu_i(\mathbf{x})
= \mathds{1}_{nm}^\top \bm\mu(\mathbf{x})
=1$,
it follows that $\bm\phi(\mathbf{x})\in\mathcal H_n$.

Following Friedkin's reflected-appraisal mechanism~\cite{friedkin2011formal}, the self-weights are updated across issues according to
\begin{equation}\label{equ:UpdateRule}
	\mathbf{x}(s+1)=\bm\phi(\mathbf{x}(s)).
\end{equation}
Whenever $\bm\phi(\mathbf{x})$ is well defined, \eqref{equ:AggregatedPower}--\eqref{equ:UpdateRule} induce a self-appraisal map on the corresponding subset of $\mathcal H_n$.

\begin{remark}
	In this paper, $\mathbf{x}(s)$ denotes the self-weight vector,
	$\bm\mu(\mathbf{x})$ the topic-level influence profile, and
	$\bm\phi(\mathbf{x})$ the aggregated agent-level update under the present aggregation rule.
	By the so-called ``social power'', we refer to the resulting agent-level influence distribution.
	Whenever the self-weight update in \eqref{equ:UpdateRule} is well defined and convergent, its limit gives the long-term agent-level social power profile.
\end{remark}

\begin{remark}
	Note that we distinguish between the simplex $\Delta_n$ of nonnegative self-weights and the affine hyperplane $\mathcal{H}_n$, the latter constituting the analytical state space for the signed self-appraisal map.
	In cooperative settings, the update is naturally confined to $\Delta_n$.
	By contrast, repelling effects may transiently drive $\mathbf{x}(s)$ outside $\Delta_n$.
	Later sections identify conditions under which the self-appraisal map is well posed on relevant domains and the asymptotic self-weight distribution lies in $\Delta_n$.
\end{remark}

\begin{remark}
	This framework extends the classical DF model by incorporating within-agent topic interdependence through $\{C_i\}_{i\in\mathcal{I}}$ together with repelling-type signed interpersonal interactions through $R$.
	It does not rely on structural balance and gauge transformations~\cite{altafini2012consensus,shi2019dynamics}, and is therefore not restricted to structurally balanced signed networks.
\end{remark}

\section{Dynamics under Shared Topic Weights}
\label{sec:baseline-sharedc}

We examine how the logic matrices $\{C_i\}_{i\in\mathcal I}$ and the relative interaction matrix $R$ shape the evolution of self-weights, with particular emphasis on the resulting equilibrium structure and long-term behavior.
In this section, the focus lies on the shared topic-weight regime, where the logic matrices admit a common normalized left eigenvector; the homogeneous case $C_i\equiv C$ is included as a special case.

Our objective here is to identify structural conditions under which the original self-appraisal map \eqref{equ:UpdateRule} associated with the coupled multi-topic DF dynamics admits an explicit reduction to a single-topic DF map on $\Delta_n$, thereby enabling a complete characterization of equilibria and long-term behavior.

\subsection{Structural Assumptions}
\label{subsec:assumptions}

We begin by introducing structural assumptions on the logic matrices and the relative interaction matrix.
These assumptions are designed to preserve interpretability of topic aggregation, guarantee regularity of the interpersonal influence dynamics, and prepare the ground for the separation analysis developed later in the section.

To maintain interpretability of agents' belief processing and cross-topic aggregation, we impose two modeling constraints.
First, following standard works on multidimensional opinion dynamics models with logical interdependence \cite{parsegov2016novel,ye2020continuous,ye2019consensus,luan2025coevolutionary}, we require each $C_i$ to possess strictly positive diagonal entries, representing self-reinforcement across topics.
Second, 
we require that all $ C_i $ share a common ``dominant'' left eigenvector, which we refer to here as the topic-weight vector and assume to be strictly positive.
These assumptions are stated below.

\begin{assumption}\label{assum:shared-c}
	For each agent $i\in\mathcal I$, the logic matrix $C_i\in\mathbb R^{m\times m}$ is ES and there exists a common vector $\bm c>\bm 0$ such that $\bm c^\top C_i=\bm c^\top \forall i\in\mathcal I$ and $\bm c^\top \mathds{1}_m=1$. Moreover, $\operatorname{diag}(C_i)>\bm 0$.
\end{assumption}

\begin{remark}\label{rmk:model-justification}
	The separation principle developed below can, in principle, be extended to the more general SPF case (allowing signed $\bm c$) or to  more general ES logic matrices with possibly negative diagonals.
	Here we impose Assumption~\ref{assum:shared-c} in order to preserve a nonnegative interpretation of aggregated topic-level influence and, consequently, of the limiting self-weight vector.
	A further implication concerns a structural restriction on signed logic matrices: as discussed  in the Appendix, negative logical interdependencies cannot exist in low dimensions ($m<3$).
	Therefore, negative logical influences require $m\ge 3$ under the present assumptions.
\end{remark}

We next impose structural conditions on the relative interaction matrix.

\begin{assumption}\label{assum:SPF-R}
	The relative interaction matrix $R\in\mathbb R^{n\times n}$ is SPF with $\bm{r}^\top R = \bm{r}^\top$ and $\bm{r}^\top \mathds{1}_{n} = 1$. Moreover, $\operatorname{diag}(R)=\bm 0$.
\end{assumption}

Assumption~\ref{assum:SPF-R} guarantees that the eigenvalue $1$ of $R$ is simple and dominant, consequently the associated normalized left eigenvector $\bm r$ is well defined.
The SPF class is broad enough to include signed relative interaction patterns beyond the ES class.
However, SPF alone is not sufficient for the subsequent DF analysis.
On the one hand, the corresponding vector $\bm r$ may still have mixed signs, which prevents a nonnegative centrality interpretation.
On the other hand, even if one restricts attention to matrices that are already ES and hence admit $\bm r>\bm 0$, the nonlinear state-dependent DF dynamics may still be unstable, with $\rho(W(\mathbf x))>1$ for some $\mathbf x$, as shown in Example~\ref{ex:counter-example-SLC} below.

To obtain a class of matrices enjoying both strict positivity of $\bm r$ and uniform stability of $W(\mathbf x)$ on $\Delta_n$, we introduce the Single-Leader Contractivity (SLC) condition below.
Based on the decomposition in \eqref{equ:DecompositionW}, define for each agent $i\in\mathcal I$
\begin{equation}\label{equ:Qi}
	Q_i:=\mathrm{diag}(\mathbf{e}_i)+(I-\mathrm{diag}(\mathbf{e}_i))R,
	\quad
	\Pi_i:=\mathds{1}_n\mathbf{e}_i^\top.
\end{equation}

\begin{assumption}[Single-Leader Contractivity]
	\label{assum:SLC-R}
	There exist a matrix $P\succ0$ and a scalar $\beta\in(0,1)$ such that $\|Q_i-\Pi_i\|_{P} \le \beta$ for all $i\in\mathcal I$.
\end{assumption}

The above SLC condition requires each single-leader matrix $Q_i$ (corresponding to agent $i$ with full self-trust) to be a strict contraction toward the rank-one projector $\Pi_i$ in a common quadratic norm.
By convexity, this leads to a uniform spectral gap for any convex mixture $W(\mathbf x)=\sum_i x_iQ_i$, which is a key ingredient for the subsequent reduction and global analysis.

\begin{proposition}[Positivity of $\bm r$ and uniform spectral gap]
	\label{prop:SLC-positivity-gap}
	Under Assumptions~\ref{assum:SPF-R} and \ref{assum:SLC-R}, the following statements hold:
	\begin{enumerate}
		\item The normalized left eigenvector $\bm r$ of $R$ is strictly positive. In particular, $R$ is ES.

		\item For $\mathbf x\in\Delta_n$, let $\bm\alpha(\mathbf x)$ be a left eigenvector of $W(\mathbf x)$ associated with eigenvalue $1$, normalized by $\bm\alpha(\mathbf x)^\top\mathds{1}_n=1$. Define $\Pi(\mathbf x):=\mathds{1}_n\bm\alpha(\mathbf x)^\top$, $S(\mathbf x):=\sum_{i=1}^n x_i(Q_i-\Pi_i)$.
		Then,
		\[
		\|S(\mathbf x)\|_P\le \beta<1, \quad \rho\bigl(W(\mathbf x)-\Pi(\mathbf x)\bigr)<1.
		\]
		In particular, the eigenvalue $1$ of $W(\mathbf x)$ is simple and dominant, and the normalized left eigenvector $\bm\alpha(\mathbf x)$ is uniquely defined.
	\end{enumerate}
\end{proposition}

\begin{proof}
	\textit{Proof of Statement 1.}
	Here and below, $\det(\cdot)$, $\ker(\cdot)$, $\operatorname{rank}(\cdot)$, and $\operatorname{span}\{\cdot\}$ denote the determinant, nullspace, matrix rank, and linear span, respectively.
	
	For each $i\in\mathcal I$, let $\widehat R_i\in\mathbb R^{(n-1)\times(n-1)}$ denote the principal submatrix of $R$ obtained by deleting the $i$-th row and the $i$-th column.
	Let $U_i$ be a permutation matrix that moves index $i$ to the first position.
	Then $R$ can be written in block form as
	\[
	U_i R U_i^\top
	=
	\begin{bmatrix}
		0 & \bm\rho_i^\top\\
		\bm\sigma_i & \widehat R_i
	\end{bmatrix},
	\]
	for some vectors $\bm\rho_i,\bm\sigma_i\in\mathbb R^{n-1}$.
	Using the definitions of $Q_i$ and $\Pi_i$, one obtains
	\[
	U_i(Q_i-\Pi_i)U_i^\top
	=
	\begin{bmatrix}
		0 & \mathbf 0^\top\\
		\bm\sigma_i-\mathds{1}_{n-1} & \widehat R_i
	\end{bmatrix}.
	\]
	Hence $U_i(Q_i-\Pi_i)U_i^\top$ is block lower triangular, and therefore
	\[
	\Lambda(Q_i-\Pi_i)=\{0\}\cup \Lambda(\widehat R_i).
	\]
	
	By Assumption~\ref{assum:SLC-R}, $\rho(Q_i-\Pi_i)\le \|Q_i-\Pi_i\|_P \le \beta<1$.
	It follows that $\rho(\widehat R_i)<1\,\forall i\in\mathcal I$.
	Consequently, every eigenvalue $\lambda_k(\widehat R_i)$ of $\widehat R_i$ satisfies $|\lambda_k(\widehat R_i)|<1$, and thus
	\[
	\det(I_{n-1}-\widehat R_i)
	=
	\prod_{k=1}^{n-1}\bigl(1-\lambda_k(\widehat R_i)\bigr)>0.
	\]
	Indeed, if $\lambda_k(\widehat R_i)\in\mathbb R$, then $|\lambda_k(\widehat R_i)|<1$ means $1-\lambda_k(\widehat R_i)>0$.
	If $\lambda\notin\mathbb R$ is an eigenvalue of $\widehat R_i$, then its complex conjugate $\bar\lambda$ is also an eigenvalue of $\widehat R_i$, and the corresponding factor is $(1-\lambda)(1-\bar\lambda)=|1-\lambda|^2>0$.

	Now define $M:=I_n-R$.
	Since $R$ is SPF, the eigenvalue $1$ of $R$ is simple.
	Equivalently, $0$ is a simple eigenvalue of $M$.
	Therefore, $\ker(M)=\operatorname{span}\{\mathds{1}_n\}$.
	By the rank-nullity theorem, $\operatorname{rank}(M)=n-1$.
	Moreover, $M\mathds{1}_n=\mathbf 0$ and $\bm r^\top M=\mathbf 0^\top$, where $\bm r^\top R=\bm r^\top$ and $\bm r^\top\mathds{1}_n=1$.
	
	Let $\operatorname{adj}(M)$ be the adjugate matrix of $M$, i.e., the transpose of the cofactor matrix.
	Recall that $M \operatorname{adj}(M)=\operatorname{adj}(M) M=\det(M)I_n$.
	Since $\det(M)=0$, this leads to $M \operatorname{adj}(M)=0$ and $\operatorname{adj}(M) M=0$.
	Hence every column of $\operatorname{adj}(M)$ belongs to $\ker(M)=\operatorname{span}\{\mathds{1}_n\}$, while every row belongs to the left nullspace, namely $\operatorname{span}\{\bm r^\top\}$.
	Because $\operatorname{rank}(M)=n-1$, at least one $(n-1)\times(n-1)$ minor of $M$ is nonzero; equivalently, $\operatorname{adj}(M)\neq 0$.
	Therefore, there exists a nonzero scalar $\gamma$ such that $\operatorname{adj}(M)=\gamma \mathds{1}_n\bm r^\top$.
	
	For each $i$, the $(i,i)$ entry of $\operatorname{adj}(M)$ equals the corresponding principal cofactor of $M$, namely
	\[
	[\operatorname{adj}(M)]_{ii}
	=
	\det(I_{n-1}-\widehat R_i).
	\]
	On the other hand, given that $\operatorname{adj}(M)=\gamma \mathds{1}_n\bm r^\top$, we also have $[\operatorname{adj}(M)]_{ii}=\gamma r_i$.
	Hence
	\[
	\det(I_{n-1}-\widehat R_i)=\gamma r_i,\quad \forall i\in\mathcal I.
	\]
	Since the left-hand side is strictly positive for any $i$, all entries $r_i$ have the same sign.
	Together with $\bm r^\top\mathds{1}_n=1$, we get $\bm r>\bm 0$.
	Therefore $R^\top\in\mathcal{PF}$, and hence $R$ is ES.

	\textit{Proof of Statement 2.} By convexity of the induced norm and Assumption~\ref{assum:SLC-R}, we have
	\begin{equation*}
		\begin{aligned}
			\|S(\mathbf x)\|_{P}\! =\! \left\|\sum_{i=1}^n x_i(Q_i-\Pi_i)\right\|_{P} \!\!\!\le\! \sum_{i=1}^n x_i\|Q_i-\Pi_i\|_{P}\le\beta<1.
		\end{aligned}
	\end{equation*}
	Define $Z(\mathbf x):=\mathds{1}_n\big(\mathbf x-\bm\alpha(\mathbf x)\big)^\top$ and note that
	\begin{align*}
		W(\mathbf x)-\Pi(\mathbf x)
		& =\sum_{i=1}^n x_i Q_i - \mathds{1}_n\bm\alpha(\mathbf x)^\top \\
		& =\sum_{i=1}^n x_i(Q_i-\Pi_i)
		+\mathds{1}_n\big(\mathbf x-\bm\alpha(\mathbf x)\big)^\top \\
		& = S(\mathbf x) + Z(\mathbf x).
	\end{align*}
	Direct calculation yields the following algebraic relations:
	\begin{itemize}
		\item $S(\mathbf x)\mathds{1}_n=\mathbf 0$ since $Q_i\mathds{1}_n=\mathds{1}_n$ and $\Pi_i\mathds{1}_n=\mathds{1}_n$ for all $i$;
		\item $Z(\mathbf x)^2= 0$ since $\sum_i(x_i-\alpha_i(\mathbf x))=0$;
		\item $S(\mathbf x)Z(\mathbf x)=0$ since $S(\mathbf x)\mathds{1}_n=\mathbf 0$.
	\end{itemize}
	Hence, for all $k\ge1$,
	\[
	(S+Z)^k = S^k + \sum_{j=0}^{k-1} S^j Z S^{k-1-j} = S^k+ZS^{k-1}.
	\]
	Taking the induced $P$-norm,
	\[
	\|(S+Z)^k\|_{P} \le \|S\|_{P}^k + \|Z\|_{P} \|S\|_{P}^{k-1} \to 0
	\quad  \text{as } k\to\infty,
	\]
	because $\|S\|_{P}\le\beta<1$. Therefore,
	$\rho\big(W(\mathbf x)-\Pi(\mathbf x)\big)<1$.
	
	Next, since $W(\mathbf x)\mathds{1}_n=\mathds{1}_n$ and
	$\bm\alpha(\mathbf x)^\top W(\mathbf x)=\bm\alpha(\mathbf x)^\top$, the rank-one matrix $\Pi(\mathbf x):=\mathds{1}_n\bm\alpha(\mathbf x)^\top$ satisfies
	\[
	\Pi(\mathbf x)^2=\Pi(\mathbf x),\  W(\mathbf x)\Pi(\mathbf x)=\Pi(\mathbf x),\  \Pi(\mathbf x)W(\mathbf x)=\Pi(\mathbf x).
	\]
	Then, $\Pi(\mathbf x)(W(\mathbf x)-\Pi(\mathbf x))=0$ and $(W(\mathbf x)-\Pi(\mathbf x))\Pi(\mathbf x)=0$.
	Hence, for every $k\ge1$,
	\[
	W(\mathbf x)^k=\Pi(\mathbf x)+\bigl(W(\mathbf x)-\Pi(\mathbf x)\bigr)^k \to \Pi(\mathbf x)
	\quad  \text{as } k\to\infty.
	\]
	This implies the eigenvalue $1$ of $W(\mathbf x)$ is simple and dominant; the normalized left eigenvector $\bm\alpha(\mathbf x)$ is uniquely defined.
\end{proof}

Proposition~\ref{prop:SLC-positivity-gap} indicates that any $R$ satisfying Assumptions~\ref{assum:SPF-R} and \ref{assum:SLC-R} is ES and yields both a strictly positive centrality vector $\bm r$ and a uniform spectral gap for $W(\mathbf x)$.
Hence, the class identified by SPF and SLC constitutes a structurally meaningful subclass of ES matrices for the DF dynamics considered here.
However, not all ES matrices satisfy Assumption~\ref{assum:SLC-R}, implying  that strict positivity of $\bm r$ alone does not guarantee the required stability, as demonstrated next.

\begin{example}\label{ex:counter-example-SLC}
	Consider a relative interaction matrix (which is ES and satisfies Assumption~\ref{assum:SPF-R}):
	\begin{equation*}
		R = \begin{bmatrix}
			0 & 1.14 & 0.78 & -0.92 \\
			0.93 & 0 & -0.35 & 0.42 \\
			-0.89 & 0.70 & 0 & 1.19 \\
			0.27 & 0.71 & 0.02 & 0 \\
		\end{bmatrix}
	\end{equation*}
	The associated dominant left eigenvector is strictly positive:
	\[
	\bm{r} = 
	\begin{bmatrix}
		0.3733 & 0.5084 & 0.1133 & 0.0050
	\end{bmatrix}^\top.
	\]
	Nevertheless, the single-leader matrix $Q_3 = \operatorname{diag}(\mathbf e_3) + (I-\operatorname{diag}(\mathbf e_3))R$	has spectral radius $\rho(Q_3) \approx 1.2249 > 1$.
	Hence $\rho(W(\mathbf{x})) > 1$ for some $\mathbf{x}$.
	Consequently, although $R$ admits a well-defined positive centrality vector $\bm r$, the system violates Assumption~\ref{assum:SLC-R}, and the corresponding DF dynamics may fail to enjoy the uniform stability needed for the subsequent analysis.
\end{example}

\subsection{The Algebraic Separation Principle}
\label{subsec:separation}

Having established regularity of the interaction component, we now use a key separability property of the coupled system.
Under the shared topic-weight assumption, the left eigenstructure of the system matrix \eqref{equ:SystemMatrix} at eigenvalue $1$ can be separated into an interaction component and a topic-weight component.
This observation leads to an explicit reduction of the original self-appraisal map on $\Delta_n$.

\begin{lemma}[Algebraic separation identity]\label{lem:sep-algebraic}
	Under Assumption~\ref{assum:shared-c}, for any $\mathbf x\in\mathcal H_n$ and any row vector $\bm\alpha^\top\in\mathbb R^{1\times n}$,
	\begin{equation}\label{eq:separation-identity}
		(\bm\alpha^\top\otimes \bm c^\top) \mathcal A(\mathbf x)
		=
		(\bm\alpha^\top W(\mathbf x))\otimes \bm c^\top.
	\end{equation}
	In particular, if one has $\bm\alpha^\top W(\mathbf x)=\bm\alpha^\top$, then $\bm\alpha^\top\otimes \bm c^\top$ is a left eigenvector of $\mathcal A(\mathbf x)$ associated with eigenvalue $1$.
\end{lemma}

\begin{proof}
	Using Assumption~\ref{assum:shared-c}, namely $\bm c^\top C_i=\bm c^\top \forall i$, one has $(\bm\alpha^\top\otimes \bm c^\top)\mathcal C_{[n]} 
	= [\alpha_1\bm c^\top C_1,\dots,\alpha_n\bm c^\top C_n] 
	= \bm\alpha^\top\otimes \bm c^\top$.
	Then,
	\[
	(\bm\alpha^\top\otimes \bm c^\top)\mathcal A(\mathbf x)
	=
	(\bm\alpha^\top\otimes \bm c^\top)(W(\mathbf x)\otimes I_m)
	=
	(\bm\alpha^\top W(\mathbf x))\otimes \bm c^\top,
	\]
	which proves \eqref{eq:separation-identity}. In particular, if
	$\bm\alpha^\top W(\mathbf x)=\bm\alpha^\top$, then $(\bm\alpha^\top\otimes \bm c^\top)\mathcal A(\mathbf x)
	=
	\bm\alpha^\top\otimes \bm c^\top$.
	Therefore $\bm\alpha^\top\otimes \bm c^\top$ is a left eigenvector of $\mathcal A(\mathbf x)$ associated with eigenvalue $1$.
\end{proof}

The above lemma identifies, for any $\mathbf x\in\Delta_n$, a natural candidate for a left eigenvector of $\mathcal A(\mathbf x)$ at eigenvalue $1$.
However, Assumption~\ref{assum:shared-c} alone does not guarantee this eigenvector is the dominant influence profile underlying the self-appraisal map.
We therefore impose the  coupled spectral regularity condition.

\begin{assumption}
	\label{assum:coupled-regularity}
	For every $\mathbf x\in\Delta_n$, the system matrix $\mathcal A(\mathbf x)$ has eigenvalue $1$ as a simple dominant eigenvalue.
\end{assumption}

\begin{remark}
	Assumption~\ref{assum:coupled-regularity} is automatically satisfied in the homogeneous case $C_i\equiv C$, where $\mathcal A(\mathbf x)=W(\mathbf x)\otimes C$.
	More generally, it still allows genuinely heterogeneous logic matrices, since the matrices $\{C_i\}_{i\in\mathcal I}$ may differ while sharing the same topic-weight vector $\bm c$.
	Its role is not to eliminate logical interdependence, but to ensure that the separated eigenvector identified in Lemma~\ref{lem:sep-algebraic} is exactly the dominant left eigenvector underlying the original self-appraisal map.
\end{remark}

Proposition~\ref{prop:SLC-positivity-gap} guarantees strict positivity of $\bm r$ and a uniform spectral gap for $W(\mathbf x)$.
Combined with Lemma~\ref{lem:sep-algebraic} and Assumption~\ref{assum:coupled-regularity}, this leads to an explicit characterization of the original self-appraisal map on $\Delta_n$ and reduces the coupled multi-topic DF dynamics to an equivalent single-topic DF map.

\begin{lemma}[Model reduction and explicit DF map]\label{lem:sep-explicit}
	Under Assumptions~\ref{assum:shared-c}--\ref{assum:coupled-regularity}, for any $\mathbf x\in\Delta_n$, the matrix $W(\mathbf x)$ admits a unique normalized dominant left eigenvector $\bm\alpha(\mathbf x)$.
	Moreover, restricting to the simplex $\Delta_n$, one has
	\begin{equation}\label{eq:sep-map}
		\alpha_i(\mathbf x)=
		\begin{cases}
			1 & \text{if } \mathbf x=\mathbf e_i \text{ for some } i,\\
			\displaystyle \dfrac{r_i}{1-x_i}\Big/\sum_{j=1}^n\dfrac{r_j}{1-x_j}
			& \text{if } \mathbf x \in \Delta_n \setminus \{\mathbf{e}_1, \dots, \mathbf{e}_n\}.
		\end{cases}
	\end{equation}
	Furthermore, $\bm\mu(\mathbf x)=\bm\alpha(\mathbf x)\otimes \bm c$ is the normalized dominant left eigenvector of $\mathcal A(\mathbf x)$,
	and the self-appraisal map simplifies to $\bm\phi(\mathbf x)=\bm\alpha(\mathbf x)$ on $\Delta_n$.
\end{lemma}

\begin{proof}
	By Proposition~\ref{prop:SLC-positivity-gap}, for any $\mathbf x\in\Delta_n$, the matrix $W(\mathbf x)$ has a unique normalized dominant left eigenvector $\bm\alpha(\mathbf x)$.

	Consider $\mathbf x\in\Delta_n\setminus\{\mathbf e_1,\dots,\mathbf e_n\}$.
	The left eigen-equation $\bm\alpha^\top W(\mathbf x)=\bm\alpha^\top$ yields $\bm\alpha^\top [\mathrm{diag}(\mathbf x)+(I-\mathrm{diag}(\mathbf x))R]=\bm\alpha^\top$. Then, $\bm\alpha^\top(I-R)=\bm\alpha^\top\mathrm{diag}(\mathbf x)(I-R)$.
	Rearranging gives
	\[
	\big(\bm\alpha^\top\mathrm{diag}(\mathds{1}_n-\mathbf x)\big)(I-R)=\bm 0^\top.
	\]
	Since $R$ is SPF and the eigenvalue $1$ is simple, the left kernel of $(I-R)$ is spanned by $\bm r^\top$, where $\bm r^\top R=\bm r^\top$ and $\bm r^\top\mathds{1}_n=1$.
	Hence there exists a scalar $\kappa\neq 0$ such that
	\[
	\bm\alpha^\top\mathrm{diag}(\mathds{1}_n-\mathbf x)=\kappa \bm r^\top,
	\]
	which implies $\alpha_i(1-x_i)=\kappa r_i,\,\forall i=1,\ldots,n$.
	Normalizing with $\sum_i\alpha_i(\mathbf x)=1$ yields
	\[
	\alpha_i(\mathbf x)=\frac{r_i/(1-x_i)}{\sum_{j=1}^n r_j/(1-x_j)},
	\]
	which is exactly \eqref{eq:sep-map} on $\Delta_n\setminus\{\mathbf e_1,\dots,\mathbf e_n\}$.

	Now consider a vertex $\mathbf x=\mathbf e_i$.
	Then $W(\mathbf e_i)=Q_i$.
	Moreover, $Q_i\mathds{1}_n=\mathds{1}_n$ and $\mathbf e_i^\top Q_i=\mathbf e_i^\top$, hence $Q_i\Pi_i=\Pi_iQ_i=\Pi_i$.
	Since $\rho(Q_i-\Pi_i)<1$ by Proposition~\ref{prop:SLC-positivity-gap}, it follows that
	\[
	Q_i^k=\Pi_i+(Q_i-\Pi_i)^k\to \Pi_i
	\quad  \text{as } k\to\infty,
	\]
	which shows that $1$ is a simple dominant eigenvalue of $Q_i$, and its unique normalized left eigenvector is $\mathbf e_i$.
	Thus $\bm\alpha(\mathbf e_i)=\mathbf e_i$, matching the vertex case in \eqref{eq:sep-map}.

	By Lemma~\ref{lem:sep-algebraic}, $\bm\alpha(\mathbf x)\otimes \bm c$ is a left eigenvector of $\mathcal A(\mathbf x)$ associated with eigenvalue $1$.
	By Assumption~\ref{assum:coupled-regularity}, the eigenvalue $1$ of $\mathcal A(\mathbf x)$ is simple and dominant, so its normalized dominant left eigenvector is uniquely defined.
	Since $(\bm\alpha(\mathbf x)^\top\otimes \bm c^\top)\mathds{1}_{nm}
	=
	(\bm\alpha(\mathbf x)^\top\mathds{1}_n)(\bm c^\top\mathds{1}_m)
	=
	1$,
	it follows that $\bm\mu(\mathbf x)=\bm\alpha(\mathbf x)\otimes \bm c$.
	
	Finally, recalling the definition of $\bm\phi$ in \eqref{equ:AggregatedPower}, we compute
	\[
	\phi_i(\mathbf x)
	=
	\sum_{p=1}^m \mu_{i,p}(\mathbf x)
	=
	\sum_{p=1}^m \alpha_i(\mathbf x)c_p
	=
	\alpha_i(\mathbf x) \bm c^\top\mathds{1}_m
	=
	\alpha_i(\mathbf x),
	\]
	where we used $\bm c^\top\mathds{1}_m=1$.
	Thus $\bm\phi(\mathbf x)=\bm\alpha(\mathbf x)$ on $\Delta_n$.
\end{proof}

\begin{remark}\label{rmk:domain-singularity}
	Although the model is defined on the hyperplane $\mathcal H_n$, Lemma~\ref{lem:sep-explicit} specializes the analysis to the simplex $\Delta_n$ for two reasons.
	First, the explicit formula \eqref{eq:sep-map} may encounter singularities outside $\Delta_n$ (e.g., when the denominator vanishes), hence the closed-form representation is not naturally defined on all of $\mathcal H_n$.
	Second, as shown in Corollary~\ref{cor:simplex}, $\Delta_n$ is forward invariant under the original self-appraisal map.
	Therefore, for trajectories initialized in $\Delta_n$, the asymptotic analysis can be carried out entirely on the simplex without loss of generality.
\end{remark}

\begin{remark}\label{rmk:n-2-trivial}
	In this paper we focus on the nontrivial case $n\ge 3$.
	Indeed, under Assumption~\ref{assum:SPF-R}, if $n=2$ then the zero-diagonal constraint forces
	$R=\big[\begin{smallmatrix} 0 & 1 \\ 1 & 0 \end{smallmatrix}\big]$, which is exactly a permutation matrix and doubly stochastic with $\bm r=\tfrac12\mathds{1}_2$.
	Substituting into \eqref{eq:sep-map} yields $\bm\alpha(\mathbf x)=\mathbf x$ for all $\mathbf x\in\Delta_2$, and hence, by Lemma~\ref{lem:sep-explicit}, the self-weight dynamics are static.
\end{remark}

\begin{corollary}[Forward invariance of the simplex]\label{cor:simplex}
	Under Assumptions~\ref{assum:shared-c}--\ref{assum:coupled-regularity}, the simplex $\Delta_n$ is forward invariant. Specifically, if $\mathbf{x}(0) \in \Delta_n$, then $\mathbf{x}(s) \in \Delta_n$ for all $s \ge 0$.
\end{corollary}

\begin{proof}
	By Proposition~\ref{prop:SLC-positivity-gap}, $\bm r>\bm 0$.
	For every $\mathbf x\in\Delta_n$, the explicit map \eqref{eq:sep-map} defines $\bm\alpha(\mathbf x)$ as a ratio of nonnegative terms and satisfies $\sum_i \alpha_i(\mathbf x)=1$, so $\bm\alpha(\mathbf x)\in\Delta_n$.
	Since Lemma~\ref{lem:sep-explicit} gives $\bm\phi(\mathbf x)=\bm\alpha(\mathbf x)$ on $\Delta_n$, the trajectory remains in $\Delta_n$.
\end{proof}

\subsection{Equilibria and Global Behavior}
\label{subsec:convergence}

We now study the global evolution of the self-weight vector $\mathbf x(s)$.
According to Lemma~\ref{lem:sep-explicit}, on $\Delta_n$ the self-appraisal map admits the explicit representation \eqref{eq:sep-map}.
Any interior equilibrium $\mathbf x^*=\bm\phi(\mathbf x^*)$ must therefore satisfy, for some scalar $\kappa$,
\begin{equation}\label{eq:FP-scalar}
	x_i^*(1-x_i^*)=\kappa r_i,\quad i=1,\ldots,n,
\end{equation}
where $\kappa=(\sum_j r_j/(1-x_j^*))^{-1}$.
For each $i$, \eqref{eq:FP-scalar} has two admissible roots: a ``small-branch'' solution $x_i^-(\kappa)\in(0,\tfrac12]$ and a ``large-branch'' solution $x_i^+(\kappa)\in(\tfrac12,1)$.
Thus, characterizing interior equilibria reduces to selecting, for each agent, one of the two branches subject to the simplex constraint $\sum_i x_i=1$.

Before classifying equilibria and establishing global convergence, we first isolate a monotonicity property that governs the ``no-interior'' regime.
It shows that once one agent's centrality reaches the threshold $1/2$, every nonvertex trajectory is pushed monotonically toward the corresponding vertex.

\begin{lemma}[Vertex attraction under the dominant centrality]
	\label{lem:vertex-attraction}
	Under Assumptions~\ref{assum:shared-c}--\ref{assum:coupled-regularity}, let $i_{\max}$ denote the unique index such that $r_{i_{\max}}=r_{\max}\ge \tfrac12$.
	Then, for every $\mathbf x\in\Delta_n\setminus\{\mathbf e_1,\dots,\mathbf e_n\}$, $\phi_{i_{\max}}(\mathbf x)>x_{i_{\max}}$.
	Consequently, for every initial state $\mathbf x(0)\in\Delta_n\setminus\{\mathbf e_1,\dots,\mathbf e_n\}$, the trajectory of \eqref{equ:UpdateRule} satisfies $\mathbf x(s)\to \mathbf e_{i_{\max}}$.
\end{lemma}

\begin{proof}
	By Proposition~\ref{prop:SLC-positivity-gap}, one has $\bm r\in\mathrm{int}\Delta_n$.
	Since $r_{i_{\max}}\ge \tfrac12$ and $\sum_i r_i=1$, the maximizer $i_{\max}$ is unique.
	
	Let $\mathbf x\in\Delta_n\setminus\{\mathbf e_1,\dots,\mathbf e_n\}$.
	By Lemma~\ref{lem:sep-explicit}, one has $\bm\phi(\mathbf x)=\bm\alpha(\mathbf x)$ and
	\[
	\phi_{i_{\max}}(\mathbf x) = \alpha_{i_{\max}}(\mathbf x) = \frac{r_{i_{\max}}/(1-x_{i_{\max}})}{\sum_{j=1}^n r_j/(1-x_j)}.
	\]
	Therefore, $\phi_{i_{\max}}(\mathbf x)>x_{i_{\max}}$ is equivalent to
	\[
	r_{i_{\max}} > x_{i_{\max}} \sum_{j\neq i_{\max}} \frac{r_j}{1-x_j}.
	\]
	For each $j\neq i_{\max}$, one has $x_j\le 1-x_{i_{\max}}$, hence $1-x_j\ge x_{i_{\max}}$.
	It follows that $x_{i_{\max}}r_j/(1-x_j)\le r_j$ for all $j\neq i_{\max}$.
	Summing over $j\neq i_{\max}$ gives
	\[
	x_{i_{\max}} \sum_{j\neq i_{\max}} \frac{r_j}{1-x_j}
	\le \sum_{j\neq i_{\max}} r_j
	= 1-r_{i_{\max}}
	\le r_{i_{\max}}.
	\]
	Moreover, the inequality is strict whenever $\mathbf x\notin\{\mathbf e_1,\dots,\mathbf e_n\}$.
	Indeed, since $n\ge 3$ and $\mathbf x\in\Delta_n\setminus\{\mathbf e_1,\dots,\mathbf e_n\}$, there exists at least one $j\neq i_{\max}$ with $x_j<1-x_{i_{\max}}$, which implies $1-x_j>x_{i_{\max}}$ and hence $x_{i_{\max}}r_j/(1-x_j)< r_j$.
	Therefore, $\phi_{i_{\max}}(\mathbf x)>x_{i_{\max}}$ for every $\mathbf x\in\Delta_n\setminus\{\mathbf e_1,\dots,\mathbf e_n\}$.
	
	Define $V(\mathbf x):=1-x_{i_{\max}}$.
	Then $V(\mathbf x)\ge 0$ on $\Delta_n$, and for every $\mathbf x\in\Delta_n\setminus\{\mathbf e_1,\dots,\mathbf e_n\}$,
	\[
	V(\bm\phi(\mathbf x)) = 1-\phi_{i_{\max}}(\mathbf x) < 1-x_{i_{\max}} = V(\mathbf x),
	\]
	which implies that $x_{i_{\max}}(s)$ is strictly increasing along every nonvertex trajectory and bounded above by $1$, so $x_{i_{\max}}(s)\to \epsilon$ for some $\epsilon\in(0,1]$.
	
	If $\epsilon<1$, choose $\varrho:=(1-\epsilon)/2>0$.
	Then $x_{i_{\max}}(s)\le 1-\varrho$ for all sufficiently large $s$, so the tail of the trajectory always lies in the compact set $\mathcal K_\varrho:=\{\mathbf x\in\Delta_n: x_{i_{\max}}\le 1-\varrho\}$,
	which does not contain $\mathbf e_{i_{\max}}$.
	Since the map $\mathbf x\mapsto \phi_{i_{\max}}(\mathbf x)-x_{i_{\max}}$ is continuous and strictly positive on $\mathcal K_\varrho$, there exists $\eta_\varrho>0$ such that $\phi_{i_{\max}}(\mathbf x)-x_{i_{\max}}\ge \eta_\varrho, \forall \mathbf x\in\mathcal K_\varrho$.
	Therefore, $x_{i_{\max}}(s+1)-x_{i_{\max}}(s)\ge \eta_\varrho$ for all sufficiently large $s$, contradicting the convergence of $x_{i_{\max}}(s)$. Hence $\epsilon=1$.
	
	Finally, since $\mathbf x(s)\in\Delta_n$ and $\sum_i x_i(s)=1$ for all $s$, the convergence $x_{i_{\max}}(s)\to 1$ implies $x_j(s)\to 0$ for all $j\neq i_{\max}$.
	Therefore, $\mathbf x(s)\to \mathbf e_{i_{\max}}$.
\end{proof}

Lemma~\ref{lem:vertex-attraction} resolves the regime in which no interior equilibrium can persist.
We now combine this monotonicity mechanism with the branch structure of \eqref{eq:FP-scalar} to obtain a complete classification of equilibria and global behavior.

\begin{theorem}[Interior equilibria and global behavior]
	\label{thm:global-final}
	Suppose that Assumptions~\ref{assum:shared-c}--\ref{assum:coupled-regularity} hold.
	Let $\bm r\in\mathrm{int}\Delta_n$ be the strictly positive dominant left eigenvector of $R$.
	Set $r_{\max}:=\max_i r_i$ and $\kappa_{\max}:=1/(4r_{\max})$.
	When the maximizer of $\bm r$ is unique, denote it by $i_{\max}$.
	Define the cumulative function for the small-branch solutions:
	\[
	\ell(\kappa):=\sum_{i=1}^n x_i^-(\kappa)
	=\sum_{i=1}^n \tfrac12 \left(1-\sqrt{1-4\kappa r_i}\right),
	\  \kappa\in(0,\kappa_{\max}].
	\]
	Then the following statements hold:
	\begin{enumerate}
		\item \textit{Boundary equilibria.}  
		The set of boundary equilibria is exactly the set of vertices $\{\mathbf e_1,\dots,\mathbf e_n\}$.
		
		\item \textit{Interior equilibria.}  
		Exactly one of the following three mutually exclusive cases occurs:
		\begin{itemize}
			\item[(a)] \textit{Small branch.}  
			If $\ell(\kappa_{\max}) \ge 1$, there exists a unique interior equilibrium
			$\mathbf x^*\in(0,\tfrac12]^n$.
			
			\item[(b)] \textit{Mixed branch.}  
			If $\ell(\kappa_{\max}) < 1$ and $r_{\max}<\tfrac12$, then necessarily the maximizer of $\bm r$ is unique; denoting it by $i_{\max}$, there exists a unique interior equilibrium $\mathbf x^*$ such that $x_{i_{\max}}^*>\tfrac12$ and $x_j^*<\tfrac12\,\forall j\neq i_{\max}$.
			
			\item[(c)] \textit{No interior equilibrium.}  
			If $r_{\max}\ge\tfrac12$, then there is no interior equilibrium.
		\end{itemize}
		\item \textit{Global convergence.}  If an interior equilibrium $\mathbf x^*$ exists (cases 2(a) and 2(b)), then every trajectory of \eqref{equ:UpdateRule} starting from $\Delta_n\setminus\{\mathbf e_1,\dots,\mathbf e_n\}$ converges to $\mathbf x^*$.
		Otherwise (i.e., case 2(c)), every such trajectory converges to the vertex $\mathbf e_{i_{\max}}$, and the maximizer $i_{\max}$ is necessarily unique.
	\end{enumerate}
\end{theorem}

\begin{proof}
	\textit{Proof of Statement 1.}
	By Lemma~\ref{lem:sep-explicit}, one has $\bm\phi(\mathbf e_i)=\bm\alpha(\mathbf e_i)=\mathbf e_i$, so each vertex is a fixed point.
	If $\mathbf x\in\partial\Delta_n$ is not a vertex, then $x_k=0$ for some $k$.
	By Proposition~\ref{prop:SLC-positivity-gap}, one has $\bm r>\bm 0$, and \eqref{eq:sep-map} yields $\alpha_k(\mathbf x)>0$.
	Hence $\phi_k(\mathbf x)=\alpha_k(\mathbf x)>0$, so $x_k\neq \phi_k(\mathbf x)$.
	Therefore $\mathbf x$ cannot be an equilibrium, and the vertices are the only boundary equilibria.
	
	\textit{Proof of Statement 2.}
	Any interior fixed point $\mathbf x\in\mathrm{int} \Delta_n$ satisfies \eqref{eq:FP-scalar}.
	Solving this quadratic equation for each component yields two admissible branches:
	\[
	x_i^{\pm}(\kappa) := \tfrac12 \left(1 \pm \sqrt{1-4\kappa r_i}\right), \quad \kappa \in (0, \kappa_{\max}].
	\]
	Thus, each $x_i \in \{x_i^-(\kappa), x_i^+(\kappa)\}$.
	Let $\mathcal{L}:=\{i: x_i=x_i^+(\kappa)\}$ denote the index set of large-branch components.
	The simplex constraint becomes
	\begin{equation}\label{eq:gE}
		g_{\mathcal{L}}(\kappa):=\sum_{i\in \mathcal{L}}x_i^+(\kappa)+\sum_{i\notin \mathcal{L}}x_i^-(\kappa)=1.
	\end{equation}
	Therefore, interior equilibria are in bijection with pairs $(\mathcal{L},\kappa)$ that solve \eqref{eq:gE}.

	Let $\mathcal M:=\{i: r_i=r_{\max}\}$ denote the set of agents with maximum centrality and let $k:=|\mathcal M|$.
	We first rule out $|\mathcal{L}|\ge 2$.
	Since $x_i^+(\kappa)\ge \tfrac12$ for all $i$ (with equality only when $i\in\mathcal M$ and $\kappa=\kappa_{\max}$),
	$g_{\mathcal{L}}(\kappa)\ge \sum_{i\in \mathcal{L}}\tfrac12\ge 1$.
	Equality $g_{\mathcal{L}}(\kappa)=1$ could only occur in the degenerate situation $\mathcal{L}=\mathcal M$, $k=2$, $\kappa=\kappa_{\max}$, and $x_j=0$ for $j\notin\mathcal M$, which does not lie in $\mathrm{int} \Delta_n$.
	Therefore, no interior equilibrium can satisfy $|\mathcal{L}|\ge 2$.
	
	Next, we analyze the case $|\mathcal{L}|=1$.
	We claim that this can occur only when $k=1$, that is, when the maximizer of $\bm r$ is unique.
	Fix $\mathcal{L}=\{i_0\}$ and define
	\[
	g_{\{i_0\}}(\kappa):=x_{i_0}^+(\kappa)+\sum_{j\neq i_0}x_j^-(\kappa),
	\quad \kappa\in(0,\kappa_{\max}].
	\]
	Its derivative is
	\[
	\frac{\mathrm d}{\mathrm d\kappa}g_{\{i_0\}}(\kappa)=\sum_{j\neq i_0} f(r_j,\kappa)-f(r_{i_0},\kappa),
	\ 
	f(r,\kappa):=\frac{r}{\sqrt{1-4\kappa r}},
	\]
	and $f(r,\kappa)$ is strictly increasing in $r$.
	Moreover,
	\[
	\lim_{\kappa\to 0^+} g_{\{i_0\}}(\kappa)
	=\lim_{\kappa\to 0^+} x_{i_0}^+(\kappa)+\sum_{j\neq i_0}\lim_{\kappa\to 0^+}x_j^-(\kappa)
	=1.
	\]
	
	If $k\ge 2$, then for any choice of $i_0$ one has $\frac{\mathrm d}{\mathrm d\kappa}g_{\{i_0\}}(\kappa)>0, \forall \kappa\in(0,\kappa_{\max}]$.
	Indeed, if $i_0\notin\mathcal M$, then there exists some $j\in\mathcal M$ such that $r_j>r_{i_0}$, which already yields strict positivity.
	If $i_0\in\mathcal M$, then there exists another index $j\in\mathcal M\setminus\{i_0\}$ with $r_j=r_{i_0}$; since $n\ge 3$ and $\bm r\in\mathrm{int}\Delta_n$, at least one additional term in the sum is strictly positive, so the derivative remains strictly positive.
	Hence $g_{\{i_0\}}(\kappa)>1, \forall \kappa\in(0,\kappa_{\max}]$,
	and thus no singleton large-branch solution exists when $k\ge 2$.
	
	Therefore, if $|\mathcal L|=1$, necessarily $k=1$.
	Let $i_{\max}$ denote the unique maximizer of $\bm r$.
	We now show that one must have $\mathcal L=\{i_{\max}\}$.
	Fix $\mathcal L=\{i_0\}$ with $i_0\neq i_{\max}$.
	Since $r_{i_{\max}}>r_{i_0}$ and $f(r,\kappa)$ is strictly increasing in $r$, one has $\frac{\mathrm d}{\mathrm d\kappa}g_{\{i_0\}}(\kappa)>0, \forall \kappa\in(0,\kappa_{\max}]$.
	Together with $\lim_{\kappa\to 0^+} g_{\{i_0\}}(\kappa)=1$, this implies $g_{\{i_0\}}(\kappa)>1, \forall \kappa\in(0,\kappa_{\max}]$,
	so no solution to \eqref{eq:gE} exists unless $\mathcal L=\{i_{\max}\}$.
	
	This leaves only two possibilities for an interior equilibrium:
	$\mathcal L=\emptyset$ (small branch), or $\mathcal L=\{i_{\max}\}$ with a unique maximizer $i_{\max}$ (mixed branch). We analyze them separately.
	
	\emph{Case (a): $\mathcal{L}=\emptyset$.}
	We have $g_\emptyset(\kappa)=\ell(\kappa)$.
	Each $x_i^-(\kappa)$ is continuous and strictly increasing on $(0,\kappa_{\max}]$ with $x_i^-(0)=0$.
	Therefore $\ell(\kappa)$ is continuous and strictly increasing with $\ell(0)=0$.
	By the intermediate value theorem, there exists a unique $\kappa^*\in(0,\kappa_{\max}]$ that satisfies $\ell(\kappa^*)=1$ if and only if $\ell(\kappa_{\max})\ge 1$.
	This yields the unique small-branch equilibrium
	$\mathbf x^*=(x_1^-(\kappa^*),\ldots,x_n^-(\kappa^*))\in(0,\tfrac12]^n$.
	
	\emph{Case (b): $\mathcal{L}=\{i_{\max}\}$ with $k=1$ and $\ell(\kappa_{\max})<1$.}
	For ease of analysis, we introduce a change of variable.
	Define $t:=\sqrt{1-4\kappa r_{\max}}\in[0,1)$, implying $\kappa=(1-t^2)/(4r_{\max})$.
	Let $\zeta_i:=r_i/r_{\max}\in(0,1]$.
	Then $\sqrt{1-4\kappa r_i} = \sqrt{1-\zeta_i(1-t^2)}$.
	Consequently, $x_i^-(\kappa)$ and $\ell(\kappa)$ can be parameterized by $t$ as:
	\begin{equation*}
		\begin{aligned}
			x_i^-(\kappa(t))&=\tfrac12 \Big(1-\sqrt{1-\zeta_i(1-t^2)}\Big) \\
			\ell(\kappa(t))&=\sum_{i=1}^n x_i^-(\kappa(t)) = \tfrac{n}{2}-\tfrac12\sum_{i=1}^n \sqrt{1-\zeta_i(1-t^2)}.
		\end{aligned}
	\end{equation*}
	Using $x_{i_{\max}}^+(\kappa)-x_{i_{\max}}^-(\kappa)=\sqrt{1-4\kappa r_{\max}}=t$, we obtain
	\[
	g_{\{i_{\max}\}}(\kappa(t))=\ell(\kappa(t))+t=:G(t),\quad t\in[0,1).
	\]
	Thus, solving $g_{\{i_{\max}\}}(\kappa)=1$ is equivalent to
	\begin{equation}\label{eq:G=1}
		G(t)=1,\quad t\in[0,1).
	\end{equation}
	Writing $h_i(t):=\sqrt{1-\zeta_i+\zeta_i t^2}$, one then derives
	$h_i''(t)=\tfrac{\zeta_i(1-\zeta_i)}{(1-\zeta_i+\zeta_i t^2)^{3/2}}\ge 0$,
	with strict inequality for $\zeta_i\in(0,1)$.
	Hence $G$ (a linear term minus a sum of convex terms) is concave on $[0,1]$.
	Moreover, $G(1)=1$ and $G(0)=\ell(\kappa_{\max})<1$.
	A direct calculation gives
	\begin{equation*}
		\begin{aligned}
			G'(1) = \Big(1-\tfrac12\sum_{i=1}^n\frac{\zeta_i t}{h_i(t)}\Big)\Big{\rvert}_{t=1} = 1-\tfrac12\sum_i \zeta_i=1-\tfrac{1}{2r_{\max}}.
		\end{aligned}
	\end{equation*}
	If $r_{\max}<\tfrac12$, then $G'(1)<0$, so for $t$ slightly to the left of $1$ we have $G(t)>1$.
	Combining $G(0)<1$ with the intermediate value theorem and concavity shows that \eqref{eq:G=1} has exactly one solution $t^*\in(0,1)$.
	This results in a unique $\kappa^*\in(0,\kappa_{\max})$ and hence the unique mixed-branch equilibrium with	$x_{i_{\max}}^*=\tfrac12(1+t^*)>\tfrac12$ and $x_j^*=x_j^-(\kappa^*)<\tfrac12$ for all $j\neq i_{\max}$.

	Finally, if $r_{\max}\ge\tfrac12$, then $G'(1)\ge 0$ and concavity implies $G(t)\le 1$ for all $t\in[0,1]$, with $G(0)<1$. Hence, no mixed-branch solution exists.
	In addition, since $\zeta_j\in(0,1)$ for $j\neq i_{\max}$, one has $\sqrt{1-\zeta_j}>1-\zeta_j$, and therefore
	\begin{equation*}
		\begin{aligned}
			\ell(\kappa_{\max})
			&=\tfrac12 \Big(1+\sum_{j\ne i_{\max}} \big(1-\sqrt{1-\zeta_j}\big)\Big) \\
			&< \tfrac12 \Big(1+\sum_{j\ne i_{\max}} \zeta_j\Big)
			=\tfrac{1}{2r_{\max}} \leq 1.
		\end{aligned}
	\end{equation*}
	Hence, no small-branch solution exists either; that is, case (c).

	\textit{Proof of Statement 3.}
	According to Corollary~\ref{cor:simplex}, the simplex $\Delta_n$ is forward invariant.
	Moreover, by Lemma~\ref{lem:sep-explicit}, the dynamics on $\Delta_n$ satisfy $\bm\phi(\mathbf x)=\bm\alpha(\mathbf x)$ and reduce exactly to the explicit single-topic DF map \eqref{eq:sep-map}.

	If an interior equilibrium $\mathbf x^*$ exists, then by Statement 2 it is unique, and case 2(c) is explicitly excluded.
	Hence $r_{\max}<\tfrac12$.
	Therefore, the classical convergence proof for the non-star DF map applies directly; see~\cite[Lem.~2.2, Thm.~4.1]{jia2015opinion}.
	Then, every trajectory with $\mathbf x(0)\in\Delta_n\setminus\{\mathbf e_1,\dots,\mathbf e_n\}$ converges to $\mathbf x^*$.
	
	If no interior equilibrium exists, then by Statement 2(c), one has $r_{\max}\ge\tfrac12$.
	Since $\bm r\in\mathrm{int}\Delta_n$, the maximizer is necessarily unique.
	In this case, Lemma~\ref{lem:vertex-attraction} implies that any trajectory with $\mathbf x(0)\in\Delta_n\setminus\{\mathbf e_1,\dots,\mathbf e_n\}$ converges to $\mathbf e_{i_{\max}}$.
\end{proof}

To visualize the branch classification in Theorem~\ref{thm:global-final}, Fig.~\ref{fig:solution_branches} plots the equilibrium branches of \eqref{eq:FP-scalar}.
The maximal centrality $r_{\max}$ determines the feasibility limit $\kappa_{\max}$: when $\kappa>\kappa_{\max}$, no real solutions to the quadratic relation exist.
Geometrically, an interior equilibrium corresponds to a vertical line $\kappa=\kappa^*$ whose branch intercepts sum to $1$.

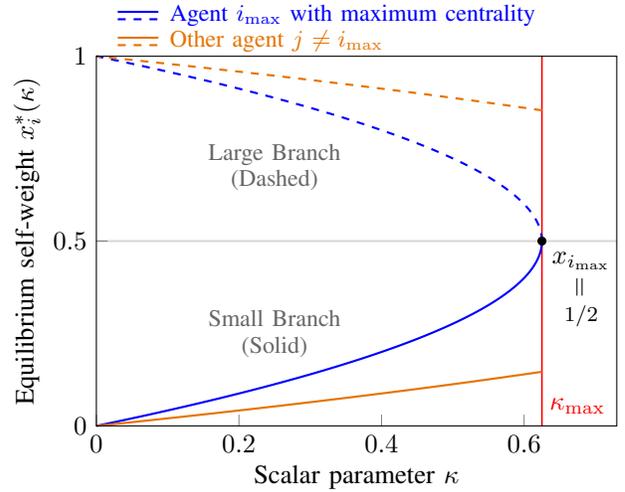
\begin{figure}[htbp]
	\begin{tikzpicture}
		\begin{axis}[
			width=8.5cm, height=6.5cm,
			xlabel={Scalar parameter $\kappa$},
			ylabel={Equilibrium self-weight $x_i^*(\kappa)$},
			xmin=0, xmax=0.73,
			ymin=0, ymax=1,
			xtick={0, 0.2, 0.4, 0.6}, 
			ytick={0, 0.5, 1},        
			axis on top=true,         
			tick pos=left,  
			ylabel style={yshift=-2pt},
			xlabel style={yshift=2pt},
			clip=false,               
			title={},
			title style={font=\small\bfseries}
			]

			\def\rmax{0.4}
			\def\rsmall{0.20}
			\pgfmathsetmacro{\kappamax}{1/(4*\rmax)}
			\pgfmathsetmacro{\ratio}{\rsmall/\rmax}
			\pgfmathsetmacro{\discrim}{sqrt(1 - \ratio)}
			\pgfmathsetmacro{\xsolidlimit}{0.5 * (1 - \discrim)}
			\pgfmathsetmacro{\xdashedlimit}{0.5 * (1 + \discrim)}
			
			\addplot[gray!30, solid, thick] coordinates {(0.00, 0.5) (0.73, 0.5)};
			\draw[red, semithick] (axis cs: \kappamax, 0) -- (axis cs: \kappamax, 1);

			\addplot[blue, thick, domain=0:0.5, samples=100] ({x*(1-x)/\rmax}, {x});
			\addplot[blue, thick, dashed, domain=0.5:1, samples=100] ({x*(1-x)/\rmax}, {x});
			
			\addplot[orange!90!black, thick, domain=0:\xsolidlimit, samples=100] ({x*(1-x)/\rsmall}, {x});
			\addplot[orange!90!black, thick, dashed, domain=\xdashedlimit:1, samples=100] ({x*(1-x)/\rsmall}, {x});

			\draw[blue, thick] (axis cs: 0.03, 1.12) -- (axis cs: 0.087, 1.12);
			\draw[blue, thick, dashed] (axis cs: 0.03, 1.10) -- (axis cs: 0.087, 1.10);
			\node[anchor=west, blue, font=\small] at (axis cs: 0.09, 1.11) {Agent $i_{\max}$ with maximum centrality};
			
			\draw[orange!90!black, thick] (axis cs: 0.03, 1.05) -- (axis cs: 0.087, 1.05);
			\draw[orange!90!black, thick, dashed] (axis cs: 0.03, 1.03) -- (axis cs: 0.087, 1.03);
			\node[anchor=west, orange!90!black, font=\small] at (axis cs: 0.09, 1.04) {Other agent $j$ $\neq i_{\max}$};

			\node[font=\small, align=center, gray!80!black] at (axis cs: 0.25, 0.70) {Large Branch\\(Dashed)};
			\node[font=\small, align=center, gray!80!black] at (axis cs: 0.25, 0.25) {Small Branch\\(Solid)};
			
			\node[red, anchor=north, font=\normalsize] at (axis cs: \kappamax+0.05, 0.1) {$\kappa_{\max}$};
			\filldraw[black] (axis cs: \kappamax, 0.5) circle (1.5pt);
			
			\node[anchor=west, font=\normalsize, black, align=center, inner sep=1pt] 
			at (axis cs: \kappamax+0.01, 0.37) {
				$x_{i_{\max}}$\\[1pt]
				\rotatebox{90}{$=$}\\[-3pt]
				{\footnotesize $1/2$}
			};
		\end{axis}
	\end{tikzpicture}
	\caption{Illustration of equilibrium branches for \eqref{eq:FP-scalar}.
		The curves describe the equilibrium self-weight $x_i^*(\kappa)$ as a function of the auxiliary scalar parameter $\kappa$.
		Blue and orange correspond to the agent $i_{\max}$ with maximum centrality and any other agent $j\neq i_{\max}$, respectively.
		For each agent, the solid segment represents the small branch ($x_i^*\le 1/2$), whereas the dashed segment represents the large branch ($x_i^*>1/2$).
		The vertical red line marks the feasibility limit $\kappa_{\max}$, at which the two branches of agent $i_{\max}$ merge at $x_{i_{\max}}^*=1/2$.}
	\label{fig:solution_branches}
\end{figure}

\begin{remark}\label{rmk:mono-computation}
	The characterization in Theorem~\ref{thm:global-final} is constructive.
	In particular, for $\mathcal{L}=\emptyset$, the root of $\ell(\kappa)=1$ can be found reliably via bisection on $(0,\kappa_{\max}]$.
	For $\mathcal{L}=\{i_{\max}\}$, the root of $G(t)=1$ can likewise be found through bisection on $(0,1)$, because $G(0)<1$ and $G(t)>1$ for $t$ sufficiently close to $1$ when $r_{\max}<\tfrac12$.
	Algorithm~\ref{alg:interior-eq} summarizes this procedure and computes the interior equilibrium whenever it exists.
\end{remark}

\begin{algorithm}[t]
	\caption{Computation of Interior Equilibria}
	\label{alg:interior-eq}
	\DontPrintSemicolon
	\SetAlgoLined
	
	\KwIn{Dominant left eigenvector $\bm r\in\mathrm{int}\Delta_n$ of $R$.}
	\KwOut{Interior equilibrium $\mathbf x^*$ or $\emptyset$ (if none exists).}
	
	Compute $r_{\max}\leftarrow \max_i r_i$ and $\kappa_{\max}\leftarrow 1/(4r_{\max})$.\;
	Compute the boundary value
	$\ell(\kappa_{\max})\leftarrow \sum_{i=1}^n \tfrac12 \left(1-\sqrt{1-r_i/r_{\max}}\right)$.\;
	
	\BlankLine
	\uIf{$\ell(\kappa_{\max})\ge 1$}{
		Find $\kappa^*\in(0,\kappa_{\max}]$ s.t. $\ell(\kappa^*)=1$ via bisection.\;
		Set the (all-small-branch) equilibrium:
		$x_i^*\leftarrow \tfrac12 \left(1-\sqrt{1-4\kappa^* r_i}\right),~\forall i$.\;
		\Return $\mathbf x^*$.\;
	}
	
	\BlankLine
	\uElseIf{$r_{\max}<\tfrac12$}{
		Identify $i_{\max}\leftarrow \arg\max_i r_i$.\;
		Define
		$G(t)\!:= t+\frac{n}{2}-\frac12\sum_{i=1}^n \sqrt{1-\frac{r_i}{r_{\max}}(1-t^2)}$.\;
		Find $t^*\in(0,1)$ s.t. $G(t^*)=1$ via bisection.\;
		Recover $\kappa^*\leftarrow (1-(t^*)^2)/(4r_{\max})$.\;
		Set the mixed-branch equilibrium: \;
		\, $x_{i_{\max}}^*\leftarrow \tfrac12(1+t^*),$ \;
		\, $x_j^*\leftarrow \tfrac12 \left(1-\sqrt{1-4\kappa^* r_j}\right),~\forall j\neq i_{\max}$.\;
		\Return $\mathbf x^*$.\;
	}
	
	\BlankLine
	\Else{
		\Return $\emptyset$.\;
	}
\end{algorithm}

Theorem~\ref{thm:global-final} establishes the existence and global convergence of the limiting self-weights.
Several immediate structural consequences follow from this classification with the fixed-point identity \eqref{eq:FP-scalar}.
We record them below.

\begin{corollary}[Order preservation]\label{cor:order-preservation}
	Suppose Assumptions~\ref{assum:shared-c}--\ref{assum:coupled-regularity} hold and an interior equilibrium $\mathbf x^*$ exists.
	Then $\mathbf x^*$ preserves the centrality ordering induced by $R$: for all $i,j\in\mathcal I$, if $r_i>r_j$ then $x_i^*>x_j^*$, and if $r_i=r_j$ then $x_i^*=x_j^*$.
\end{corollary}

\begin{proof}
	By Theorem~\ref{thm:global-final}, the interior equilibrium $\mathbf x^*$ is unique.
	Moreover, \eqref{eq:FP-scalar} gives $x_i^*(1-x_i^*)=\kappa^* r_i, i=1,\dots,n$.
	Define $f(z):=z(1-z)$, which is strictly increasing on $(0,\tfrac12]$.
	
	In the small-branch case, $x_i^*\in(0,\tfrac12]\,\forall i$.
	One therefore has $r_i>r_j \iff f(x_i^*)>f(x_j^*) \iff x_i^*>x_j^*$, and similarly $r_i=r_j \iff f(x_i^*)=f(x_j^*) \iff x_i^*=x_j^*$.
	
	In the mixed-branch case, there exists a unique $i_{\max}$ such that $x_{i_{\max}}^*>\tfrac12$ and $x_j^*<\tfrac12\,\forall j\neq i_{\max}$, and necessarily $r_{i_{\max}}=\max_k r_k>r_j\,\forall j\neq i_{\max}$.
	For any $j,k\neq i_{\max}$, both $x_j^*$ and $x_k^*$ lie in $(0,\tfrac12)$, so the same monotonicity argument leads to $r_j>r_k \iff x_j^*>x_k^*$ and $r_j=r_k \iff x_j^*=x_k^*$.
	Finally, $x_{i_{\max}}^*>\tfrac12>x_j^*\,\forall j\neq i_{\max}$, consistent with $r_{i_{\max}}>r_j$.
\end{proof}

A direct consequence of order preservation is a characterization of democratic configurations, in which self-weights are uniformly distributed.

\begin{corollary}[Democracy]\label{cor:democracy}
	With Assumptions~\ref{assum:shared-c}--\ref{assum:coupled-regularity}, the interior equilibrium is $\mathbf x^*=\tfrac{1}{n}\mathds{1}$ if and only if the centrality vector is uniform, i.e., $\bm r=\tfrac{1}{n}\mathds{1}$.
\end{corollary}

\begin{proof}
	($\Rightarrow$) If $\mathbf x^*=\tfrac{1}{n}\mathds{1}$, then \eqref{eq:FP-scalar} gives
	$\kappa^* r_i=x_i^*(1-x_i^*)=(n-1)/n^2$, which is constant for all $i$.
	Since $\kappa^*>0$ and $\bm r^\top\mathds{1}=1$, we obtain $\bm r=\tfrac{1}{n}\mathds{1}$.

	($\Leftarrow$) If $\bm r=\tfrac{1}{n}\mathds{1}$, then $r_{\max}=1/n$, and the uniform vector $\mathbf x^*=\tfrac{1}{n}\mathds{1}$ satisfies $x_i^*(1-x_i^*)=(n-1)/n^2$, which matches \eqref{eq:FP-scalar} with $\kappa=(n-1)/n$.
	Moreover, $\ell(\kappa_{\max})=\sum_i \tfrac12(1-\sqrt{1-n/n})=n/2>1$. Hence, by Theorem~\ref{thm:global-final}(2a), a unique interior equilibrium exists, which must be $\mathbf x^*=\tfrac{1}{n}\mathds{1}$.
\end{proof}

While democracy corresponds to a perfectly balanced power distribution, the small-branch equilibrium more generally represents a pluralistic regime in which no single agent holds a majority of the power, i.e., $x_i^*\le \tfrac12$ for all $i$.
We next provide a simple explicit sufficient condition ensuring that this regime occurs when the maximizer is unique.

\begin{corollary}[An explicit sufficient condition for pluralism]\label{cor:sufficient-small}
	Under the hypotheses of Theorem~\ref{thm:global-final}, assume the maximizer of $\bm r$ is unique, and let $r_{\max}$ denote its maximum entry.
	If $r_{\max}\le \frac{n-1}{3n-4}$,
	then the small-branch interior equilibrium exists.
\end{corollary}

\begin{proof}
	By Theorem~\ref{thm:global-final}, the small-branch equilibrium exists whenever $\ell(\kappa_{\max})\ge 1$.
	Let $i_{\max}$ denote the unique maximizer, and write $\zeta_j:=r_j/r_{\max}\,\forall j\neq i_{\max}$.
	Since $f(z):=1-\sqrt{1-z}$ is convex on $[0,1]$, by Jensen's inequality,
	\[
	\ell(\kappa_{\max})
	=\tfrac12\Big(1+\sum_{j\neq i_{\max}} f(\zeta_j)\Big)
	\ge \tfrac12\Big(1+(n-1)f(\bar\zeta)\Big),
	\]
	where $\bar\zeta:=\frac{1}{n-1}\sum_{j\neq i_{\max}}\zeta_j
	=\frac{1-r_{\max}}{(n-1)r_{\max}}$.
	Hence $\ell(\kappa_{\max})\ge 1$ is ensured by $(n-1)f(\bar\zeta)\ge 1$, i.e.,
	$\sqrt{1-\bar\zeta}\le \frac{n-2}{n-1}$.
	Squaring and simplifying yields $r_{\max}\le \frac{n-1}{3n-4}$.
\end{proof}

The previous corollaries concern the existence pattern and structural location of equilibria.
We conclude this section by recording a further local stability property of interior equilibria in the separable regime.
Although the argument employs the Jacobian estimates developed in the Appendix, it applies to all interior equilibria classified by Theorem~\ref{thm:global-final}.

\begin{corollary}[Local attractivity of interior equilibria in the separable regime]
	\label{cor:local-attractivity-separable}
	Suppose that Assumptions~\ref{assum:shared-c}--\ref{assum:coupled-regularity} hold.
	If the self-appraisal map admits an interior equilibrium $\mathbf x^*$, then $\mathbf x^*$ is locally attractive in a norm equivalent to $\ell_1$.
\end{corollary}

\begin{proof}
	By Lemma~\ref{lem:sep-explicit}, the self-appraisal map satisfies $\bm\phi(\mathbf x)=\bm\alpha(\mathbf x)$ on $\Delta_n$.
	If the small-branch equilibrium satisfies $x_i^*<\tfrac12$ for all $i$, then by Lemma~\ref{lem:jacobian-compact},
	\[
	J(\mathbf x^*)=\kappa^*\bigl(\operatorname{diag}(\mathbf a^*)-\mathbf x^*(\mathbf a^*)^\top\bigr),
	\quad
	a_j^*=\frac{r_j}{(1-x_j^*)^2}.
	\]
	Using the equilibrium identity $\kappa^* a_j^*=x_j^*/(1-x_j^*)$, the $j$-th absolute column sum of $J(\mathbf x^*)$ is
	\[
	\sum_{i=1}^n |J_{ij}(\mathbf x^*)|
	=
	\kappa^* a_j^*
	\Bigl((1-x_j^*)+\sum_{i\neq j}x_i^*\Bigr)
	=
	2x_j^*.
	\]
	Hence $\|J(\mathbf x^*)\|_1 = 2\max_j x_j^* <1$,
	so $\mathbf x^*$ is locally attractive in the $\ell_1$-norm.
	
	If instead the small-branch equilibrium satisfies $x_i^*=\tfrac12$ for some $i$, the conclusion follows from Lemma~\ref{lem:weighted-L1}.
	If $\mathbf x^*$ is the mixed-branch equilibrium, the conclusion follows from Lemma~\ref{lem:diag-weight-mixed}.
	In all these cases, there exists a norm equivalent to $\ell_1$ under which the map is locally contractive at $\mathbf x^*$, so $\mathbf x^*$ is locally attractive.
\end{proof}

\section{Dynamics Beyond Shared Topic Weights}
\label{sec:beyond-separable}

The analysis in Section~\ref{sec:baseline-sharedc} builds on Assumption~\ref{assum:shared-c}, under which the coupled system admits an exact algebraic reduction to a single-topic DF map.
We now move beyond this shared topic-weight regime.
The purpose of this section is twofold: to identify a local robustness property near a separable reference profile, and to elucidate a structural limitation of the present self-appraisal aggregation rule under heterogeneous logic.
In particular, Section~\ref{subsec:robustness} develops the perturbation analysis around a separable reference profile, while Section~\ref{subsec:aggregation-invariance} shows what survives, and what is lost, after aggregation at the agent level.

\subsection{Local Robustness under Small Logic Perturbations}
\label{subsec:robustness}

We begin with small perturbations around a separable reference logic profile.
Once the shared topic-weight condition is lost, the separated representation from Section~\ref{sec:baseline-sharedc} is in general no longer available, so the analysis must return to the coupled system matrix $\mathcal A(\mathbf x)$ itself.
Our aim here is to quantify how the eigenvalue $1$ and its associated left eigenstructure vary under small perturbations, and then to derive bounds for the induced update map from these estimates.

Consider $\{\bar C_i\}_{i=1}^n$ as a reference logic profile satisfying the shared topic-weight condition in Assumption~\ref{assum:shared-c}, with common topic-weight vector $\bar{\bm c}>\bm 0$, and $\{C_i\}_{i=1}^n$ as a perturbed logic profile.
Denote $\bar{\mathcal C}_{[n]}:=\mathrm{blkdiag}(\bar C_1,\dots,\bar C_n)$ and $\bar{\mathcal A}(\mathbf x):=\bar{\mathcal C}_{[n]}(W(\mathbf x)\otimes I_m)$.
The following analysis is carried out relative to a separable reference system whose coupled eigenvalue at $1$ remains uniformly isolated over $\Delta_n$.

\begin{assumption}
	\label{assum:robust-ref}
	The reference logic profile $\{\bar C_i\}_{i=1}^n$ and the relative interaction matrix $R$ satisfy the following conditions:
	\begin{enumerate}
		\item For each $i\in\mathcal I$, $\bar C_i$ is ES with
		$\bar{\bm c}^\top \bar C_i=\bar{\bm c}^\top$,
		$\bar{\bm c}^\top\mathds{1}_m=1$,
		and $\operatorname{diag}(\bar C_i)>\bm 0$.
		Moreover, all eigenvalues of $\bar C_i$ other than $1$ satisfy
		$|\lambda_k(\bar C_i)|\le 1-\delta_C$
		for some $\delta_C\in(0,1]$.
		
		\item The matrix $R$ satisfies Assumptions~\ref{assum:SPF-R} and \ref{assum:SLC-R}.
		In particular, following Proposition~\ref{prop:SLC-positivity-gap}, the non-dominant eigenvalues of $W(\mathbf x)$ satisfy
		$|\lambda_k(W(\mathbf x))|\le 1-\delta_W$
		uniformly over $\mathbf x\in\Delta_n$, for some $\delta_W\in(0,1]$.
	\end{enumerate}
\end{assumption}

\begin{assumption}
	\label{assum:robust-gap}
	For every $\mathbf x\in\Delta_n$, the system matrix $\bar{\mathcal A}(\mathbf x)$ has a simple eigenvalue $1$, and all its remaining eigenvalues satisfy $|\lambda_k(\bar{\mathcal A}(\mathbf x))|\le 1-\delta_A$ for some $\delta_A\in(0,1]$.
\end{assumption}

\begin{remark}
	In the homogeneous case $\bar C_i\equiv \bar C$, one obtains $\bar{\mathcal A}(\mathbf x)=W(\mathbf x)\otimes \bar C$.
	Therefore Assumption~\ref{assum:robust-gap} follows directly from the spectral gaps of $W(\mathbf x)$ and $\bar C$, and one may take $\delta_A=\min\{\delta_C,\delta_W\}$.
	When the matrices $\bar C_i$ are not identical, $\bar{\mathcal A}(\mathbf x)$ is no longer a Kronecker product, even though the reference profile remains separable.
	Assumption~\ref{assum:robust-gap} is then imposed as the reference-profile counterpart of Assumption~\ref{assum:coupled-regularity}, guaranteeing the uniform spectral isolation needed for the Riesz-projector perturbation analysis below.
\end{remark}

The perturbed logic matrices $\{C_i\}_{i=1}^n$ here are also assumed to satisfy the same basic regularity requirements, namely each $C_i$ is ES with strictly positive diagonal entries.
Consequently, $\mathcal A(\mathbf x)$ always retains the eigenvalue $1$ on $\Delta_n$.
To compare the perturbed and reference systems uniformly over $\Delta_n$, fix the contour $\Gamma:=\{z\in\mathbb C: |z-1|=\rho\}$ with $\rho:=\delta_A/2$; denote $K_{\rm res}:=\sup_{\mathbf x\in\Delta_n, z\in\Gamma}\|(zI-\bar{\mathcal A}(\mathbf x))^{-1}\|_1 < \infty$.

\begin{assumption}\label{assum:robust-small}
	Define $E(\mathbf x):=(\mathcal C_{[n]}-\bar{\mathcal C}_{[n]})(W(\mathbf x)\otimes I_m)$.
	Assume that there exists a sufficiently small $\varepsilon>0$ such that $\|E(\mathbf x)\|_1 \le \varepsilon$ for every $\mathbf x\in\Delta_n$, and $K_{\rm res}\varepsilon \le \frac12$.
\end{assumption}

Under Assumption~\ref{assum:robust-small}, the perturbed and reference systems can be compared on the same contour $\Gamma$.
We first obtain a uniform resolvent bound, then compare the associated projectors and normalized left eigenvectors, and finally translate these estimates to the induced update map.

\begin{lemma}[Uniform resolvent bounds]\label{lem:robust-res}
	Under Assumptions~\ref{assum:robust-ref}--\ref{assum:robust-small}, the resolvents of the coupled reference system $\bar{\mathcal A}(\mathbf x)$ and the perturbed system $\mathcal A(\mathbf x)$ are uniformly bounded on $\Gamma$.
	In particular, defining $K_{\rm res}^+:=\sup_{\mathbf x\in\Delta_n, z\in\Gamma}\|(zI-\mathcal A(\mathbf x))^{-1}\|_1$, one has $ K_{\rm res}^+ \le 2K_{\rm res} $.
\end{lemma}

\begin{proof}
	The finiteness of $K_{\rm res}$ follows directly from Assumption~\ref{assum:robust-gap}. That is, for each $\mathbf x\in\Delta_n$, the contour $\Gamma$ encloses the simple eigenvalue $1$ of $\bar{\mathcal A}(\mathbf x)$ and is uniformly separated from the remaining spectrum.
	Hence $zI-\bar{\mathcal A}(\mathbf x)$ is invertible for all $(\mathbf x,z)\in\Delta_n\times\Gamma$.
	Since $(\mathbf x,z)\mapsto (zI-\bar{\mathcal A}(\mathbf x))^{-1}$ is continuous and $\Delta_n\times\Gamma$ is compact, the supremum defining $K_{\rm res}$ is finite.
	
	For the perturbed resolvent, write $zI-\mathcal A = (zI-\bar{\mathcal A})(I-(zI-\bar{\mathcal A})^{-1}E)$.
	By Assumption~\ref{assum:robust-small}, $\|(zI-\bar{\mathcal A})^{-1}E\|_1 \le K_{\rm res}\varepsilon \le \frac12$.
	Therefore $I-(zI-\bar{\mathcal A})^{-1}E$ is invertible, and the Neumann-series estimate gives
	\[
	\|(zI-\mathcal A)^{-1}\|_1
	\le \frac{\|(zI-\bar{\mathcal A})^{-1}\|_1}{1-\|(zI-\bar{\mathcal A})^{-1}E\|_1}
	\le 2K_{\rm res}.
	\]
	Taking the supremum yields $K_{\rm res}^+\le 2K_{\rm res}$.
\end{proof}

With both resolvents controlled on the same contour $\Gamma$, the corresponding Riesz projectors can be compared uniformly.

\begin{lemma}[Eigenvector perturbation bound]\label{lem:rankone-proj}
	Let $\bar{\mathcal{P}}(\mathbf x)$ and $\mathcal{P}(\mathbf x)$ be the Riesz spectral projectors of $\bar{\mathcal A}(\mathbf x)$ and ${\mathcal A}(\mathbf x)$, respectively, associated with the isolated eigenvalue $1$ (which is enclosed by $\Gamma$).
	Under Assumptions~\ref{assum:robust-ref}--\ref{assum:robust-small}, there exist constants $K_{\rm proj}, K_{\mu}>0$ such that:
	\begin{enumerate}
		\item The projector perturbation satisfies
		\begin{equation}\label{eq:Pdiff}
			\|\mathcal{P}(\mathbf x)-\bar{\mathcal{P}}(\mathbf x)\|_1 \le K_{\rm proj} \varepsilon,
			\quad \forall \mathbf x\in\Delta_n.
		\end{equation}
		\item Let $\bm{\mu}(\mathbf x)$ and $\bar{\bm\mu}(\mathbf x):=\bm\alpha(\mathbf x)\otimes\bar{\bm c}$ be left eigenvectors of ${\mathcal A}(\mathbf x)$ and $\bar{\mathcal A}(\mathbf x)$, respectively, associated with eigenvalue $1$, satisfying $\bm\mu(\mathbf x)^\top \mathds{1}_{nm}=1$ and $\bar{\bm\mu}(\mathbf x)^\top \mathds{1}_{nm}=1$.
		Then
		\begin{equation}\label{eq:vecdiff}
			\|\bm\mu(\mathbf x)-\bar{\bm\mu}(\mathbf x)\|_1
			\le K_{\mu} \varepsilon,
			\quad \forall \mathbf x\in\Delta_n.
		\end{equation}
	\end{enumerate}
\end{lemma}

\begin{proof}
	Using the contour integral identity
	\[
	\mathcal{P}-\bar{\mathcal{P}}
	=\frac{1}{2\pi i}\oint_{\Gamma}(zI-\bar{\mathcal A})^{-1}E(zI-\mathcal A)^{-1} \mathrm{d}z,
	\]
	together with Lemma~\ref{lem:robust-res}, we obtain
	\begin{equation*}
		\begin{aligned}
			\|\mathcal{P}-\bar{\mathcal{P}}\|_1 
			&\le \frac{|\Gamma|}{2\pi}
			\sup_{z\in\Gamma}\|(zI-\bar{\mathcal{A}})^{-1}\|_1 \|E\|_1
			\sup_{z\in\Gamma}\|(zI-\mathcal{A})^{-1}\|_1 \\
			&\le \frac{|\Gamma|}{2\pi} K_{\rm res} K_{\rm res}^+ \varepsilon.
		\end{aligned}
	\end{equation*}
	Since $|\Gamma|=2\pi\rho$, this proves \eqref{eq:Pdiff} with, e.g., $K_{\rm proj}:=2\rho K_{\rm res}^2$.
	
	Next, because $C_i\mathds{1}_m=\mathds{1}_m$ for all $i$ and $W(\mathbf x)\mathds{1}_n=\mathds{1}_n$, one has $\mathcal A(\mathbf x)\mathds{1}_{nm}=\mathds{1}_{nm}$, so $1\in\Lambda(\mathcal A(\mathbf x))$ for every $\mathbf x\in\Delta_n$.
	By Lemma~\ref{lem:robust-res}, $\Gamma$ lies in the resolvent set of $\mathcal A(\mathbf x)$, hence the spectrum inside $\Gamma$ is isolated.
	Moreover, Assumption~\ref{assum:robust-gap} implies $\bar{\mathcal P}(\mathbf x)$ has one-dimensional range.
	Hence, for sufficiently small $\varepsilon$, the spectral subspace dimension inside $\Gamma$ is preserved under perturbation (see, e.g., \cite[Ch.~II, Thm.~5.1]{kato2013perturbation}). Therefore $\mathcal P(\mathbf x)$ also has one-dimensional range.
	
	Let ${\bm u}^\top:=\bm\mu(\mathbf x)^\top$ and $\bar{\bm u}^\top:=\bar{\bm\mu}(\mathbf x)^\top$. 
	Then one has ${\bm u}^\top \mathcal{P} = {\bm u}^\top $ and $\bar{\bm u}^\top \bar{\mathcal{P}} = \bar{\bm u}^\top$.
	Define ${\bm v}^\top:={\bm u}^\top\bar{\mathcal{P}}$.
	One then derives ${\bm u}^\top-{\bm v}^\top
	={\bm u}^\top(I-\bar{\mathcal{P}})
	={\bm u}^\top(\mathcal{P}-\bar{\mathcal{P}})$,
	implying $\|{\bm u}-{\bm v}\|_1\le \|\mathcal{P}-\bar{\mathcal{P}}\|_1$.
	Since $\bar{\mathcal P}(\mathbf x)$ has one-dimensional range, there exists a scalar $\sigma=\sigma(\mathbf x)$ such that ${\bm v}^\top=\sigma \bar{\bm u}^\top$.
	Using ${\bm u}^\top\mathds{1}_{nm}=1$ and $\bar{\bm u}^\top\mathds{1}_{nm}=1$, one obtains $|\sigma-1| = |({\bm v}^\top - {\bm u}^\top)\mathds{1}_{nm}| \le \|{\bm u}-{\bm v}\|_1 \|\mathds{1}_{nm}\|_\infty = \|{\bm u}-{\bm v}\|_1$.
	Therefore,
	\[
	\|{\bm u}-\bar{\bm u}\|_1
	\le \|{\bm u}-{\bm v}\|_1+\|{\bm v}-\bar{\bm u}\|_1
	=\|{\bm u}-{\bm v}\|_1+|\sigma-1| \|\bar{\bm u}\|_1.
	\]
	Because $\bar{\bm u}=\bm\alpha(\mathbf x)\otimes\bar{\bm c}$ with $\bm\alpha(\mathbf x)\in\Delta_n$, $\bar{\bm c}\in\mathrm{int}\Delta_m$, and $\bar{\bm c}^\top\mathds{1}_m=1$, one has $\|\bar{\bm u}\|_1=1$.
	Hence
	\[
	\|{\bm u}-\bar{\bm u}\|_1
	\le 2\|{\bm u}-{\bm v}\|_1
	\le 2\|\mathcal P-\bar{\mathcal P}\|_1
	\le 2K_{\rm proj}\varepsilon,
	\]
	which proves \eqref{eq:vecdiff} with $K_\mu:=2K_{\rm proj}$.
\end{proof}

Lemma~\ref{lem:rankone-proj} shows that the normalized left eigenvector associated with the isolated eigenvalue $1$ varies at most linearly with the heterogeneity level $\varepsilon$.
Since the aggregation in \eqref{equ:AggregatedPower} is linear, this directly yields a perturbation bound for the induced update map.

\begin{proposition}[Uniform perturbation bound for the induced update map]\label{prop:map-deviation}
	Under Assumptions~\ref{assum:robust-ref}--\ref{assum:robust-small}, there exists a constant $K_{\rm map}>0$ such that the induced update map $\bm\phi(\mathbf x)$ and the separable map $\bm\alpha(\mathbf x)$ satisfy
	\[
	\|\bm\phi(\mathbf x)-\bm\alpha(\mathbf x)\|_1 \le \Psi_{\rm map}, \quad \forall \mathbf x\in\Delta_n,
	\]
	where $\Psi_{\rm map}:=K_{\rm map}\varepsilon$.
\end{proposition}

\begin{proof}
	Recall from \eqref{equ:AggregatedPower} that the aggregation is linear:
	$\bm\phi(\mathbf x)=B\bm\mu(\mathbf x)$.
	For the reference system, define
	$\bar{\bm\phi}(\mathbf x):=B\bar{\bm\mu}(\mathbf x)=\bm\alpha(\mathbf x)$,
	where $B:=I_n\otimes \mathds{1}_m^\top \in \mathbb{R}^{n\times nm}$.
	Because $\|B\|_1=1$, Lemma~\ref{lem:rankone-proj} gives 
	\[
	\|\bm\phi(\mathbf x)-\bm\alpha(\mathbf x)\|_1
	\le \|\bm\mu(\mathbf x)-\bar{\bm\mu}(\mathbf x)\|_1
	\le K_\mu \varepsilon.
	\]
	The claim follows with $K_{\rm map}:=K_\mu$.
\end{proof}

\begin{remark}
	In this subsection, $\bm\phi$ denotes the update induced by the normalized left eigenvector of $\mathcal A(\mathbf x)$ associated with the eigenvalue $1$.
	When eigenvalue $1$ is simple dominant, this coincides with the original self-appraisal update in Section~\ref{sec:ModelFormulation}.
\end{remark}

\begin{remark}
	Proposition~\ref{prop:map-deviation} provides a perturbative bound obtained directly from Lemma~\ref{lem:rankone-proj}.
	For the DF model in the paper, this estimate is generally not sharp; Section~\ref{subsec:aggregation-invariance} shows that, whenever the original self-appraisal update is well defined, a stronger structural identification is available.
\end{remark}

Proposition~\ref{prop:map-deviation} also yields a boundary persistence property away from the vertices.

\begin{lemma}[Boundary behavior away from vertices]
	\label{lem:bdry-away}
	Under Assumptions~\ref{assum:robust-ref}--\ref{assum:robust-small}, one has
	$\bm\phi(\mathbf e_i)=\mathbf e_i, \forall i\in\mathcal I$.
	Moreover, for any $\eta>0$, define $\mathcal K_{\partial,\eta}:=
	\{ \mathbf x\in\partial\Delta_n: \min_{1\le i\le n}\|\mathbf x-\mathbf e_i\|_1\ge \eta \}$.
	Then, $\mathcal K_{\partial,\eta}$ is compact, and there exists $\varepsilon_\eta>0$ such that, whenever $0<\varepsilon\le \varepsilon_\eta$, one has $\bm\phi(\mathbf x)\in\mathrm{int}\Delta_n, \forall \mathbf x\in \mathcal K_{\partial,\eta}$.
\end{lemma}

\begin{proof}
	Fix $i\in\{1,\dots,n\}$.
	For $\mathbf x=\mathbf e_i$, one has $W(\mathbf e_i)=Q_i$ and
	$\mathcal A(\mathbf e_i)=\mathcal C_{[n]}(Q_i\otimes I_m)$.
	Let $\bm c_i$ denote the normalized left eigenvector of $C_i$ at eigenvalue $1$, satisfying
	$\bm c_i^\top C_i=\bm c_i^\top$ and $\bm c_i^\top\mathds{1}_m=1$.
	Define $\bm v^\top:=\mathbf e_i^\top\otimes \bm c_i^\top$.
	Then
	\[
	\bm v^\top\mathcal A(\mathbf e_i)
	=
	(\mathbf e_i^\top\otimes \bm c_i^\top)\mathcal C_{[n]}(Q_i\otimes I_m)
	=
	(\mathbf e_i^\top Q_i)\otimes \bm c_i^\top
	=
	\bm v^\top.
	\]
	Hence, $\bm v^\top$ is a left eigenvector of $\mathcal A(\mathbf e_i)$ associated with eigenvalue $1$.
	Since Lemma~\ref{lem:rankone-proj} ensures that the spectral subspace of $\mathcal A(\mathbf e_i)$ enclosed by $\Gamma$ is one-dimensional, the above eigenvector is unique up to scaling after normalization.
	By \eqref{equ:AggregatedPower}, $\bm\phi(\mathbf e_i)=\mathbf e_i$.
	
	Now fix $\eta>0$.
	The set $\mathcal K_{\partial,\eta}$ is closed in $\partial\Delta_n$ and stays a positive distance away from the vertices, so it is compact.
	For any $\mathbf x\in \mathcal K_{\partial,\eta}$, one has $\mathbf x\in\partial\Delta_n\setminus\{\mathbf e_1,\dots,\mathbf e_n\}$.
	By Lemma~\ref{lem:sep-explicit}, applied to the reference profile under Assumptions~\ref{assum:robust-ref} and \ref{assum:robust-gap}, the reference map satisfies $\bar{\bm\phi}(\mathbf x)=\bm\alpha(\mathbf x)\in\mathrm{int}\Delta_n$ for every $\mathbf x\in \mathcal K_{\partial,\eta}$.
	Because each component $\alpha_i(\mathbf x)$ is continuous and strictly positive on $\mathcal K_{\partial,\eta}$, we may define
	\[
	\underline{\alpha}_{\partial,\eta}
	:=
	\min_{\mathbf x\in \mathcal K_{\partial,\eta}}
	\min_{i\in\mathcal I}\alpha_i(\mathbf x)>0.
	\]
	By Proposition~\ref{prop:map-deviation}, $\|\bm\phi(\mathbf x)-\bm\alpha(\mathbf x)\|_1\le \Psi_{\rm map}$ uniformly on $\Delta_n$.
	Choose $\varepsilon_\eta>0$ such that $\Psi_{\rm map}<\underline{\alpha}_{\partial,\eta}$ whenever $0<\varepsilon\le\varepsilon_\eta$. Then for every $\mathbf x\in \mathcal K_{\partial,\eta}$ and every $i\in\mathcal I$,
	\[
	\phi_i(\mathbf x)
	\ge
	\alpha_i(\mathbf x)-|\phi_i(\mathbf x)-\alpha_i(\mathbf x)|
	>
	0.
	\]
	Hence $\bm\phi(\mathbf x)\in\mathrm{int}\Delta_n$ for all $\mathbf x\in \mathcal K_{\partial,\eta}$.
\end{proof}

The perturbation bound in Proposition~\ref{prop:map-deviation} also yields persistence of locally attractive interior equilibria near the separable reference map.

\begin{theorem}[Persistence of locally attractive interior equilibria]
	\label{thm:robust-local}
	Let Assumptions~\ref{assum:robust-ref}--\ref{assum:robust-small} hold.
	Suppose that the reference map $\bar{\bm\phi}=\bm\alpha$ admits an interior equilibrium $\mathbf x^*\in\mathrm{int}\Delta_n$ and is a strict contraction on
	\[
	\mathcal N_{\star}:=\{\mathbf z\in\Delta_n:\|\mathbf z-\mathbf x^*\|_{\star}\le r\}\subset\mathrm{int}\Delta_n
	\]
	with respect to some norm $\|\cdot\|_{\star}$ on $\mathbb R^n$, with contraction factor
	$q:=\mathrm{Lip}_{\mathcal N_{\star}}^{\star}(\bm\alpha)<1$.
	Let $\nu_{\star}>0$ be such that $\|\mathbf v\|_{\star}\le \nu_{\star}\|\mathbf v\|_1$ for all $\mathbf v\in\mathbb R^n$.
	If $\nu_{\star}\Psi_{\rm map}\le (1-q)r$, then the induced update map $\bm\phi$ admits a fixed point $\mathbf x_\varepsilon^*\in\mathcal N_{\star}$ such that
	\[
	\|\mathbf x_\varepsilon^*-\mathbf x^*\|_{\star}
	\le
	\frac{\nu_{\star}\Psi_{\rm map}}{1-q}.
	\]
	Moreover, if the perturbation is small enough that
	$q+K_{\rm lip}^{\star}\widetilde\varepsilon<1$,
	where $\widetilde\varepsilon$ is the perturbation quantity appearing in Lemma~\ref{lem:lip-blackbox},
	then $\mathbf x_\varepsilon^*$ is the unique locally attractive fixed point of $\bm\phi$ in $\mathcal N_{\star}$.
\end{theorem}

\begin{proof}
	We first show that $\bm\phi(\mathcal N_{\star})\subseteq\mathcal N_{\star}$.
	For any $\mathbf x\in\mathcal N_{\star}$,
	\begin{align*}
		\|\bm\phi(\mathbf x)-\mathbf x^*\|_{\star}
		&\le
		\|\bm\phi(\mathbf x)-\bm\alpha(\mathbf x)\|_{\star} 
		+
		\|\bm\alpha(\mathbf x)-\bm\alpha(\mathbf x^*)\|_{\star} \\
		&\le
		\nu_{\star}\Psi_{\rm map}+q\|\mathbf x-\mathbf x^*\|_{\star} \le (1-q)r+qr=r.
	\end{align*}
	Therefore $\bm\phi$ maps the compact convex set $\mathcal N_{\star}$ into itself, and Brouwer's fixed point theorem yields a fixed point $\mathbf x_\varepsilon^*\in\mathcal N_{\star}$.
	
	To estimate its displacement, use $\bm\phi(\mathbf x_\varepsilon^*)=\mathbf x_\varepsilon^*$:
	\begin{align*}
		\|\mathbf x_\varepsilon^*-\mathbf x^*\|_\star
		&=
		\|\bm\phi(\mathbf x_\varepsilon^*)-\bm\alpha(\mathbf x^*)\|_\star\\
		&\le
		\|\bm\phi(\mathbf x_\varepsilon^*)-\bm\alpha(\mathbf x_\varepsilon^*)\|_\star
		+
		\|\bm\alpha(\mathbf x_\varepsilon^*)-\bm\alpha(\mathbf x^*)\|_\star\\
		&\le
		\nu_{\star}\Psi_{\rm map}+q\|\mathbf x_\varepsilon^*-\mathbf x^*\|_{\star}.
	\end{align*}
	Hence $(1-q)\|\mathbf x_\varepsilon^*-\mathbf x^*\|_{\star} \le \nu_{\star}\Psi_{\rm map}$, giving the claimed bound.
	
	By Lemma~\ref{lem:lip-blackbox}, for sufficiently small perturbations one has
	\[
	\mathrm{Lip}_{\mathcal N_{\star}}^{\star}(\bm\phi)
	\le
	\mathrm{Lip}_{\mathcal N_{\star}}^{\star}(\bm\alpha)+K_{\rm lip}^{\star}\widetilde\varepsilon
	=
	q+K_{\rm lip}^{\star}\widetilde\varepsilon
	<1.
	\]
	Therefore $\bm\phi$ is a contraction on $\mathcal N_{\star}$ in the norm $\|\cdot\|_{\star}$.
	Banach's fixed point theorem then implies $\mathbf x_\varepsilon^*$ is the unique fixed point of $\bm\phi$ in $\mathcal N_{\star}$ and is locally attractive.
\end{proof}

The preceding theorem treats the case in which the separable reference map admits a locally attractive interior equilibrium.
Together with Lemma~\ref{lem:bdry-away}, it yields a complementary statement for the regime in which no interior equilibrium exists.

\begin{corollary}[Equilibrium consequences near a separable reference profile]
	\label{cor:robust-dichotomy}
	Let Assumptions~\ref{assum:robust-ref}--\ref{assum:robust-small} hold.
	Then the following statements hold for sufficiently small perturbations:
	\begin{enumerate}
		\item If the reference map $\bm\alpha$ admits a locally attractive interior equilibrium $\mathbf x^*$, then the induced update map $\bm\phi$ admits a unique locally attractive interior equilibrium $\mathbf x_\varepsilon^*$ in a neighborhood of $\mathbf x^*$, and $\mathbf x_\varepsilon^*\to \mathbf x^*$ as $\varepsilon\to 0$.
		
		\item If the reference map $\bm\alpha$ has no interior equilibrium, then for every $\eta>0$ there is $\varepsilon_\eta>0$ such that, whenever $0<\varepsilon\le \varepsilon_\eta$, the induced update map $\bm\phi$ has no boundary fixed point in $\mathcal K_{\partial,\eta}:=
		\{ \mathbf x\in\partial\Delta_n: \min_{1\le i\le n}\|\mathbf x-\mathbf e_i\|_1\ge \eta \}$.
		Moreover, for any sequence $\varepsilon_\ell\to 0$, if the corresponding induced maps $\bm\phi^{(\ell)}$ admit interior equilibria, every limit point of $\{\mathbf x^{(\ell)}\}$ belongs to the vertex set $\{\mathbf e_1,\dots,\mathbf e_n\}$.
	\end{enumerate}
\end{corollary}

\begin{proof}
	\textit{Proof of Statement 1.}
	By Corollary~\ref{cor:local-attractivity-separable}, whenever the reference map admits an interior equilibrium $\mathbf x^*$, there exist a norm $\|\cdot\|_{\star}$ equivalent to $\ell_1$ and also a closed neighborhood $\mathcal N_{\star}\subset\mathrm{int}\Delta_n$ of $\mathbf x^*$ on which $\bm\alpha$ is a strict contraction.
	Hence Theorem~\ref{thm:robust-local} applies and yields a unique locally attractive fixed point $\mathbf x_\varepsilon^*$ of $\bm\phi$ in $\mathcal N_{\star}$, with $\mathbf x_\varepsilon^*\to\mathbf x^*$ as $\varepsilon\to0$.

	\textit{Proof of Statement 2.}
	The first claim follows from Lemma~\ref{lem:bdry-away}: for any fixed $\eta>0$, sufficiently small perturbations map $\mathcal K_{\partial,\eta}$ into $\mathrm{int}\Delta_n$, so no point of $\mathcal K_{\partial,\eta}$ can be a boundary fixed point.

	We now prove the statement on the limit points of interior equilibria.
	Let $\{\varepsilon_\ell\}_{\ell\ge 1}$ be any sequence such that $\varepsilon_\ell\to 0$, and suppose that for each $\ell$, the induced update map $\bm\phi^{(\ell)}$ admits an interior equilibrium $\mathbf x^{(\ell)}\in\mathrm{int}\Delta_n$.
	Since $\Delta_n$ is compact, there exists a subsequence $\{\mathbf x^{(\ell_j)}\}$ and a point $\bar{\mathbf x}\in\Delta_n$ such that $\mathbf x^{(\ell_j)}\to \bar{\mathbf x}$.
	Because each $\mathbf x^{(\ell_j)}$ is a fixed point of $\bm\phi^{(\ell_j)}$, one has $\mathbf x^{(\ell_j)}=\bm\phi^{(\ell_j)}(\mathbf x^{(\ell_j)})$.
	By Proposition~\ref{prop:map-deviation},
	\[
	\|\bm\phi^{(\ell_j)}(\mathbf x^{(\ell_j)})-\bm\alpha(\mathbf x^{(\ell_j)})\|_1
	\le
	K_{\rm map}\varepsilon_{\ell_j}\to 0.
	\]
	Combining this with continuity of $\bm\alpha$ yields
	\[
	\bar{\mathbf x}
	=
	\lim_{j\to\infty}\mathbf x^{(\ell_j)}
	=
	\lim_{j\to\infty}\bm\phi^{(\ell_j)}(\mathbf x^{(\ell_j)})
	=
	\lim_{j\to\infty}\bm\alpha(\mathbf x^{(\ell_j)})
	=
	\bm\alpha(\bar{\mathbf x}),
	\]
	so $\bar{\mathbf x}$ is a fixed point of $\bm\alpha$.
	Since $\bm\alpha$ has no interior equilibrium by hypothesis, $\bar{\mathbf x}\in\partial\Delta_n$.
	By Theorem~\ref{thm:global-final}(1), the only boundary equilibria of $\bm\alpha$ are the vertices.
	Hence $\bar{\mathbf x}\in\{\mathbf e_1,\dots,\mathbf e_n\}$.
	Since the convergent subsequence was arbitrary, every limit point of $\{\mathbf x^{(\ell)}\}$ belongs to the vertex set $\{\mathbf e_1,\dots,\mathbf e_n\}$.
\end{proof}

\begin{remark}
	\label{rmk:robust-no-interior}
	In the no-interior regime of the reference map, the perturbation results provide only a partial persistence statement for the induced update map.
	They show that boundary fixed points cannot persist away from small neighborhoods of the vertices, and any interior equilibria appearing under perturbation must approach the vertex set as the perturbation level tends to zero.
	We do not claim here a full global convergence classification under perturbations, since the explicit scalar reduction in Section~\ref{sec:baseline-sharedc} is no longer available once the shared topic-weight condition is lost.
	Establishing such a result would require additional dynamics analysis near the vertices.
\end{remark}

\subsection{Aggregation Invariance beyond Shared Topic Weights}
\label{subsec:aggregation-invariance}

Section~\ref{subsec:robustness} quantifies how small logic perturbations affect the coupled eigenstructure near a separable reference system.
We now turn to the agent level.
Once the self-appraisal update is well defined in the sense of Section~\ref{sec:ModelFormulation}, how much of this heterogeneity remains visible after the aggregation in \eqref{equ:AggregatedPower}?

Throughout this subsection, logic matrices remain admissible in the sense of the model, namely ES with strictly positive diagonal entries.
For the invariance result itself, however, the only algebraic property needed is $C_i\mathds{1}_m=\mathds{1}_m$ for all $i\in\mathcal I$.

\begin{proposition}[Aggregation invariance]
	\label{prop:aggregation-invariance}
	Suppose Assumptions~\ref{assum:SPF-R} and~\ref{assum:SLC-R} hold.
	Let $\mathbf x\in\Delta_n$ be such that $\mathcal A(\mathbf x)$ has $1$ as a simple dominant eigenvalue, and let $\bm\mu(\mathbf x)$ be the corresponding normalized left eigenvector satisfying $\bm\mu(\mathbf x)^\top\mathds{1}_{nm}=1$.
	If $C_i\mathds{1}_m=\mathds{1}_m$ for all $i\in\mathcal I$,
	then the aggregated self-appraisal update rule defined by \eqref{equ:AggregatedPower} satisfies $\bm\phi(\mathbf x)=\bm\alpha(\mathbf x)$,
	where $\bm\alpha(\mathbf x)$ is the unique normalized left eigenvector of $W(\mathbf x)$ associated with eigenvalue $1$.
\end{proposition}

\begin{proof}
	Let $T:=I_n\otimes\mathds{1}_m\in\mathbb R^{nm\times n}$.
	Since $C_i\mathds{1}_m=\mathds{1}_m$ for all $i\in\mathcal I$, one has $\mathcal C_{[n]}T=T$.
	Furthermore, $(W(\mathbf x)\otimes I_m)T
	=
	(W(\mathbf x)\otimes I_m)(I_n\otimes\mathds{1}_m)
	=
	W(\mathbf x)\otimes\mathds{1}_m
	=
	TW(\mathbf x)$.
	Hence,
	\[
	\mathcal A(\mathbf x)T
	=
	\mathcal C_{[n]}(W(\mathbf x)\otimes I_m)T
	=
	\mathcal C_{[n]}TW(\mathbf x)
	=
	TW(\mathbf x).
	\]
	By \eqref{equ:AggregatedPower}, $\bm\phi(\mathbf x)^\top=\bm\mu(\mathbf x)^\top T$.
	It follows that
	\begin{align*}
		\bm\phi(\mathbf x)^\top W(\mathbf x)
		& =
		\bm\mu(\mathbf x)^\top T W(\mathbf x) \\
		& =
		\bm\mu(\mathbf x)^\top \mathcal A(\mathbf x)T
		=
		\bm\mu(\mathbf x)^\top T
		=
		\bm\phi(\mathbf x)^\top.
	\end{align*}
	Moreover, $\bm\phi(\mathbf x)^\top\mathds{1}_n
	=
	\bm\mu(\mathbf x)^\top T\mathds{1}_n
	=
	\bm\mu(\mathbf x)^\top\mathds{1}_{nm}
	=
	1$.
	Thus, $\bm\phi(\mathbf x)$ is a normalized left eigenvector of $W(\mathbf x)$ associated with eigenvalue $1$.
	By Proposition~\ref{prop:SLC-positivity-gap}, this eigenvector is unique, so $\bm\phi(\mathbf x)=\bm\alpha(\mathbf x)$.
\end{proof}

Proposition~\ref{prop:aggregation-invariance} shows that, after aggregation, the update $\bm\phi(\mathbf x)$ coincides exactly with the unique normalized left eigenvector $\bm\alpha(\mathbf x)$ of the matrix $W(\mathbf x)$ at eigenvalue $1$.
Accordingly, whenever the original self-appraisal update is well defined, the present aggregation rule implies that the aggregated agent-level map is determined solely by the interaction component.

This does not mean that heterogeneity across logic matrices becomes trivial at the topic level.
The following construction isolates a topic-level mechanism occurring before aggregation, and shows that substantial transfer disparities may still arise across heterogeneous logic matrices.

Consider a structured heterogeneous profile with one distinguished broker (i.e., an agent whose logic matrix has broad connections across topics), indexed by $n$, and a nonempty set of specialists (i.e., agents with logic matrices concentrated on a single topic) $\mathcal I_{\mathrm{sp}}\subseteq \mathcal I\setminus\{n\}$.
Each specialist $i\in\mathcal I_{\mathrm{sp}}$ is assigned a focal topic $\sigma(i)\in\{1,\dots,m\}$.
Here, the map $\sigma$ need not be injective or surjective; thus, multiple specialists may share the same focal topic, and some topics may have no specialist.
Let $\bar\chi\in(0,1)$ be fixed.
For each $i\in\mathcal I_{\mathrm{sp}}$, suppose that the logic matrix is parameterized by $\chi\in(0,\bar\chi]$ as
\[
C_i(\chi)
=
(1-\chi)\mathds{1}_m\mathbf e_{\sigma(i)}^\top+\chi \widehat C_i,
\]
where $\widehat C_i\in\mathbb R^{m\times m}$ is independent of $\chi$, satisfies $\widehat C_i\mathds{1}_m=\mathds{1}_m$, and is chosen so that $C_i(\chi)$ remains an admissible logic matrix for every $\chi\in(0,\bar\chi]$.
Here, $\bar\chi<1$ guarantees that the specialist family remains concentrated on the focal topic uniformly over the admissible range of $\chi$.
For the broker, let
\[
C_n=\mathds{1}_m\mathbf b^\top, \quad \mathbf b = [b_1,\ldots,b_m]^\top \in\mathrm{int}\Delta_m.
\]
In addition, any remaining agents in $\mathcal I\setminus(\mathcal I_{\mathrm{sp}}\cup\{n\})$ may have arbitrary admissible logic matrices, as they play no role in the estimate below.

The next proposition isolates a disparity generated solely by the logic matrices before the self-appraisal aggregation map.
Here, $\mathbf e_j$ denotes the $j$-th topic basis vector, while $\mathbf e_{\sigma(i)}^\top$ selects the component corresponding to the specialist's focal topic.
It offers a topic-level contrast to Proposition~\ref{prop:aggregation-invariance} at the agent level.
	
\begin{proposition}[Direct vs.\ broker-mediated topic-level transfer gains]
	\label{prop:gain-separation}
	Under the structured heterogeneous profile introduced above, for any specialist $i\in\mathcal I_{\mathrm{sp}}$ and any topic index $j\neq \sigma(i)$, one has, as $\chi\to 0$,
	\begin{enumerate}
		\item \textit{Direct gain:} $|\mathbf e_{\sigma(i)}^\top C_i(\chi)\mathbf e_j|=O(\chi)$.
		
		\item \textit{Broker-mediated gain:} $|\mathbf e_{\sigma(i)}^\top C_i(\chi)C_n\mathbf e_j| = b_j = \Theta(1)$.
	\end{enumerate}
\end{proposition}

\begin{proof}
	Fix $i\in\mathcal I_{\mathrm{sp}}$ and $j\neq \sigma(i)$.
	By the definition of $C_i(\chi)$,
	\begin{align*}
		\mathbf e_{\sigma(i)}^\top C_i(\chi)\mathbf e_j
		& =
		(1-\chi)\mathbf e_{\sigma(i)}^\top\mathds{1}_m\mathbf e_{\sigma(i)}^\top\mathbf e_j
		+
		\chi \mathbf e_{\sigma(i)}^\top \widehat C_i \mathbf e_j \\
		& =
		\chi \mathbf e_{\sigma(i)}^\top \widehat C_i \mathbf e_j.
	\end{align*}
	Hence, $|\mathbf e_{\sigma(i)}^\top C_i(\chi)\mathbf e_j| = \chi |\mathbf e_{\sigma(i)}^\top \widehat C_i \mathbf e_j| \le M_i \chi$,
	where $M_i:=\|\widehat C_i\|_1$ is independent of $\chi$.
	Then, $|\mathbf e_{\sigma(i)}^\top C_i(\chi)\mathbf e_j|=O(\chi)$.
	
	On the other hand, since $C_n\mathbf e_j=b_j\mathds{1}_m$ and $C_i(\chi)\mathds{1}_m=\mathds{1}_m$,
	one gets
	\[
	\mathbf e_{\sigma(i)}^\top C_i(\chi)C_n\mathbf e_j
	=
	b_j \mathbf e_{\sigma(i)}^\top C_i(\chi)\mathds{1}_m
	=
	b_j \mathbf e_{\sigma(i)}^\top\mathds{1}_m
	=
	b_j.
	\]
	Because $\mathbf b\in\mathrm{int}\Delta_m$ and thus $b_j>0$ (independent of $\chi$), one has $|\mathbf e_{\sigma(i)}^\top C_i(\chi)C_n\mathbf e_j| = b_j = \Theta(1)$.
\end{proof}

Combined with Proposition~\ref{prop:aggregation-invariance}, this yields the key contrast of this subsection: under the present self-appraisal aggregation rule, these topic-level disparities do not alter the aggregated agent-level map.
Hence, the limitation lies in the aggregation mechanism rather than in the absence of topic-level heterogeneity itself.
Making heterogeneous logic visible at the agent level would in general require an aggregation rule sensitive to the topic composition of each agent's profile, rather than only to its total mass.

\section{Numerical Simulations}
\label{sec:Simulations}

\begin{figure*}[t!]
	\centering
	\begin{subfigure}[b]{0.23\textwidth}
		\centering
		\includegraphics[width=\linewidth]{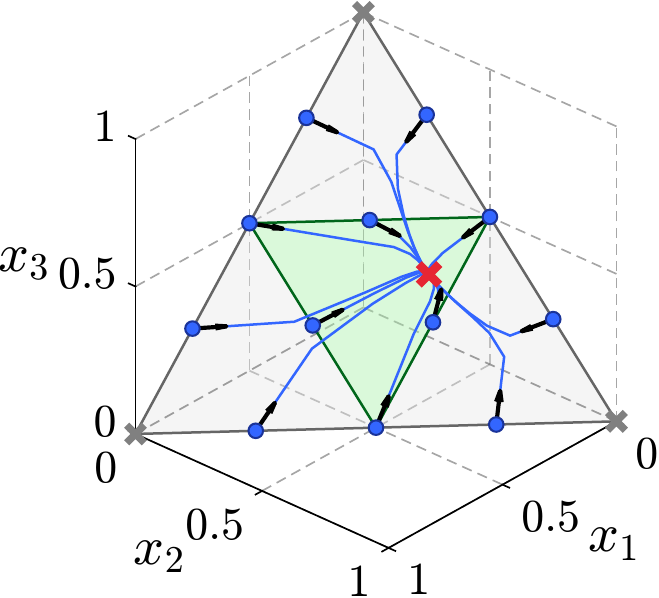}
		\caption{$\bm{r} = [0.249, 0.387, 0.364]^\top$}
		\label{fig:simplex_small}
	\end{subfigure}
	\hfill
	\begin{subfigure}[b]{0.23\textwidth}
		\centering
		\includegraphics[width=\linewidth]{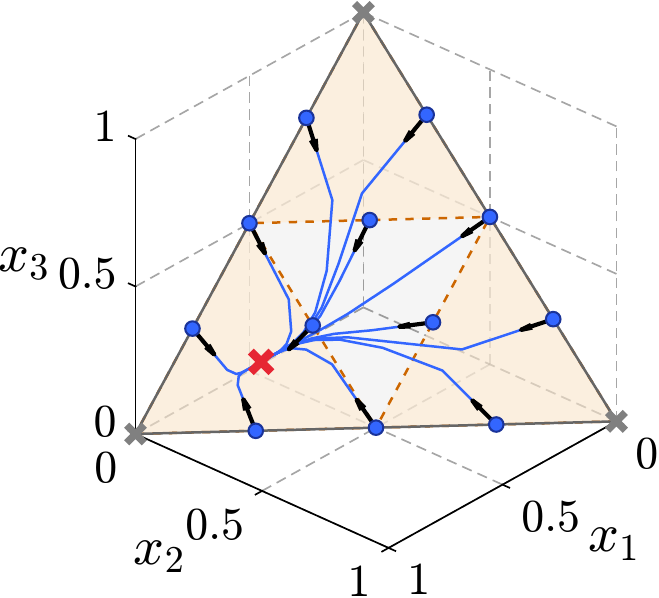}
		\caption{$\bm{r} = [0.441, 0.290, 0.269]^\top$}
		\label{fig:simplex_mixed}
	\end{subfigure}
	\hfill
	\begin{subfigure}[b]{0.23\textwidth}
		\centering
		\includegraphics[width=\linewidth]{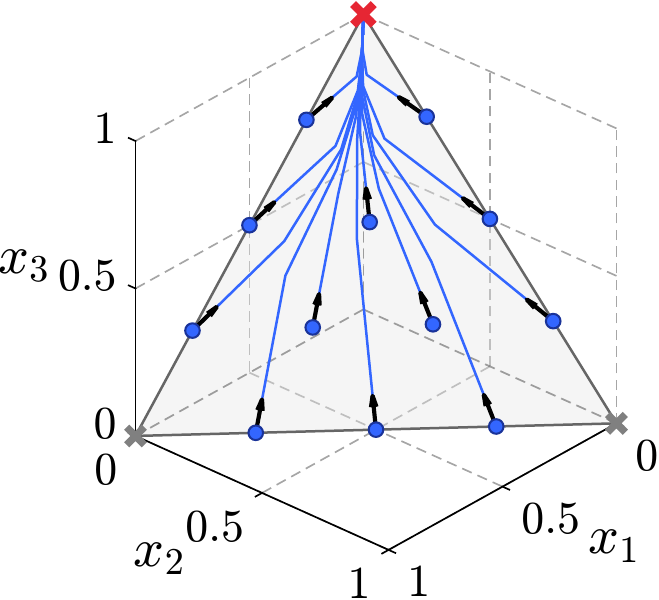}
		\caption{$\bm{r} = [0.205, 0.167, 0.628]^\top$}
		\label{fig:simplex_no_interior}
	\end{subfigure}
	\hfill
	\begin{subfigure}[b]{0.23\textwidth}
		\centering
		\includegraphics[width=\linewidth]{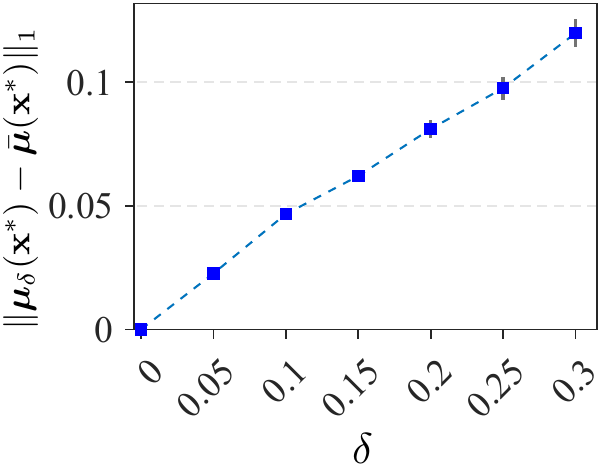}
		\caption{Perturbation near baseline (a)}
		\label{fig:robustness}
	\end{subfigure}
	\caption{Numerical illustration of the baseline dynamics and perturbation response.
		Panels (a)--(c): Evolution of self-weights on the 3-agent simplex $\Delta_3$ under a baseline system in the shared topic-weight regime.
		The trajectories (blue lines) start from multiple non-vertex initial states (blue circles) and converge to the stable equilibrium (red cross). The grey crosses denote the unstable vertex equilibria.
		Panel (a) shows the small-branch regime, in which all trajectories converge to a unique pluralistic interior equilibrium.
		Panel (b) shows the mixed-branch regime, in which all trajectories converge to a unique interior equilibrium with $x^*_{i_{\max}} > 1/2$.
		Panel (c) shows the no-interior regime, in which the dynamics collapse to the vertex $\mathbf e_3$ corresponding to the unique dominant agent.
		Panel (d) reports the perturbation response near the small-branch baseline in panel (a): for each perturbation level $\delta$, it plots the mean value of $\|\bm\mu_\delta(\mathbf x^*)-\bar{\bm\mu}(\mathbf x^*)\|_1$ at the reference interior equilibrium $\mathbf x^*$ over 500 independent trials, with vertical bars showing the 95\% confidence interval. 
        } 
	\label{fig:sim_all}
\end{figure*}

This section presents numerical simulations for a network of $n=3$ agents discussing $m=3$ topics to illustrate:
(i) the equilibrium classification and convergence properties in Theorem~\ref{thm:global-final};
(ii) the eigenvector perturbation behavior underlying the analysis in Section~\ref{subsec:robustness}.
We first construct a baseline system in the shared topic-weight regime satisfying the assumptions used in Section~\ref{sec:baseline-sharedc}, and then synthesize three representative relative interaction matrices whose normalized left eigenvectors $\bm r$ fall into the three mutually exclusive regimes in Theorem~\ref{thm:global-final}.

Fig.~\ref{fig:sim_all}(a)--(c) illustrate the baseline dynamics in Section~\ref{sec:baseline-sharedc}.
These trajectories start from multiple non-vertex initial conditions and demonstrate the global convergence behavior across the three regimes in Theorem~\ref{thm:global-final}.
In Figs.~\ref{fig:sim_all}(a) and \ref{fig:sim_all}(b), all trajectories converge to the unique interior equilibrium computed via Algorithm~\ref{alg:interior-eq}, whereas in Fig.~\ref{fig:sim_all}(c), the dynamics collapse to the vertex $\mathbf e_3$, corresponding to the unique dominant agent.

Fig.~\ref{fig:sim_all}(d) illustrates the eigenvector perturbation behavior analyzed in Section~\ref{subsec:robustness}.
We employ the stable small-branch instance from Fig.~\ref{fig:sim_all}(a) as the reference system and introduce random heterogeneous perturbations to the logic matrices, with magnitude indexed numerically by $\delta$.
For each perturbation level $\delta$, we evaluate the normalized left eigenvector $\bm\mu_\delta(\mathbf x^*)$ of the perturbed coupled system matrix at the reference interior equilibrium $\mathbf x^*$, and compare it with its reference counterpart $\bar{\bm\mu}(\mathbf x^*)$.
The observed near-linear growth of $\|\bm\mu_\delta(\mathbf x^*)-\bar{\bm\mu}(\mathbf x^*)\|_1$ as $\delta$ increases is consistent with the linear perturbation estimate in Lemma~\ref{lem:rankone-proj}.

\section{Concluding Remarks}
\label{sec:Conclusion}

In this paper we studied DF dynamics over signed influence networks with interdependent topics.
Despite the analytical difficulties introduced by repelling antagonistic interactions, the resulting self-appraisal dynamics admit a rigorous characterization under suitable structural conditions.
In the shared topic-weight regime, the model reduces exactly to an explicit scalar DF map, yielding a complete classification of limiting self-weight configurations together with global convergence on the simplex.
In this regime, the limiting self-weight vector represents the corresponding long-term agent-level social power profile, so interaction centrality remains the main determinant of long-term social power formation.

Beyond exact separability, the analysis further clarifies the scope of logic heterogeneity in the present framework.
Small heterogeneous perturbations preserve locally attractive interior equilibria, demonstrating that the centrality-based description of social power is locally robust.
At the same time, substantial topic-level effects generated by the heterogeneous logic do not alter the aggregated self-appraisal map.
These findings indicate that, in the current model, the main obstruction to transmitting topic-level logic heterogeneity to long-term agent-level social power lies in the aggregation mechanism.

Future research naturally points in two directions.
One is to weaken the current spectral and contractivity requirements while maintaining well-posed signed self-appraisal dynamics.
The other is to design topic-sensitive aggregation mechanisms under which heterogeneous logic can genuinely reshape agent-level social power formation.

\appendix

\textit{Proof of Remark~\ref{rmk:model-justification}.}
It suffices to consider the case $m=2$. Let $C\in\mathbb{R}^{2\times 2}$ be an ES matrix with positive diagonal entries:
\[
C=
\begin{bmatrix}
	\gamma & 1-\gamma\\
	\tau & 1-\tau
\end{bmatrix},
\]
where $\gamma>0$ and $1-\tau>0$ by hypothesis.
Its eigenvalues are $\lambda_1=1$ and $\lambda_2=\operatorname{tr}(C)-1=\gamma-\tau$.
Since $C$ is ES with $\rho(C)=1$, we must have $|\lambda_2|=|\gamma-\tau|<1$.
Moreover, the normalized left Perron eigenvector $\bm\pi^\top$ satisfying
$\bm\pi^\top C=\bm\pi^\top$ and $\bm\pi^\top\mathds{1}=1$
must be strictly positive.

Assume, for contradiction, that $C$ contains a negative entry.
Because the diagonal entries are positive, any negative entry must be off-diagonal, which leads to the following two cases:

\emph{Case 1: $1-\gamma<0$, i.e., $\gamma>1$.}
Since $\tau<1$ and $|\gamma-\tau|<1$, we have $\gamma-\tau<1$, hence
$\tau>\gamma-1>0$.
Thus $C$ has sign pattern
$\bigl[\begin{smallmatrix}+&-\\+&+\end{smallmatrix}\bigr]$.
Let $\bm\pi=[\pi_1,\pi_2]^\top$.
From $\bm\pi^\top C=\bm\pi^\top$, the first component equation is $(\gamma-1)\pi_1+\tau\pi_2=0$.
Since $\gamma-1>0$ and $\tau>0$, this equality forces $\pi_1$ and $\pi_2$ to have opposite signs, contradicting $\bm\pi>\bm 0$.

\emph{Case 2: $\tau<0$.}
Again, from $(\gamma-1)\pi_1+\tau\pi_2=0$, and using $\tau<0$ together with $\pi_1,\pi_2>0$, we obtain $\gamma-1>0$, i.e., $\gamma>1$.
It then follows that $ \lambda_2=\gamma-\tau>1$ (because $-\tau>0$), which implies $\rho(C)>1$, contradicting the ES property.

Therefore, no negative entries appear in a logic matrix when $m=2$ under the ES and positive-diagonal assumptions.
Hence negative logical influences require $m\ge 3$. \hfill $\blacksquare$

\begin{lemma}\label{lem:jacobian-compact}
	For any $\mathbf x\in\mathrm{int}\Delta_n$, denote $u_i:=r_i/(1-x_i)>0$, $S:=\sum_{j=1}^n u_j$, and $\kappa:=1/S$.
	Then the DF map \eqref{eq:sep-map} satisfies $\alpha_i(\mathbf x) = u_i/S = \kappa r_i / (1-x_i)$, and its Jacobian is
	\[
	J(\mathbf x):=\frac{\partial\bm\alpha}{\partial\mathbf x}=\kappa\bigl(\operatorname{diag}(\mathbf a)-\bm\alpha \mathbf a^\top\bigr),
	\  \text{with } a_i:=\frac{r_i}{(1-x_i)^2}.
	\]
\end{lemma}

\begin{proof}
	By the quotient rule,
	\(
	\partial\alpha_i/\partial x_k
	=\big((\partial u_i/\partial x_k)S-u_i(\partial S/\partial x_k)\big)/S^2.
	\)
	We have $\partial u_i/\partial x_k=\delta_{ik} r_i/(1-x_i)^2=\delta_{ik}a_i$ and $\partial S/\partial x_k=r_k/(1-x_k)^2=a_k$.
	Substituting these identities, together with $\kappa=1/S$ and $\alpha_i=u_i/S$, yields
	\[
	\frac{\partial\alpha_i}{\partial x_k}
	=\kappa^2\big(\delta_{ik}a_i S-u_i a_k\big)
	=\kappa\big(\delta_{ik}a_i-\alpha_i a_k\big),
	\]
	which is exactly the $(i,k)$-th entry of $\kappa(\mathrm{diag}(\mathbf a)-\bm\alpha \mathbf a^\top)$.
\end{proof}

\begin{lemma}[Diagonal $\ell_1$-contraction at a small-branch equilibrium]\label{lem:weighted-L1}
	Let $\mathbf x^*\in\mathrm{int}\Delta_n$ be the small-branch equilibrium.
	If the boundary-degenerate case $\mathcal M:=\{j: x_j^*=\tfrac12\}\neq\emptyset$ occurs, then there exists a diagonal matrix $\Xi=\mathrm{diag}(\xi_1,\dots,\xi_n)\succ0$ such that
	\(
	\|\Xi^{-1}J(\mathbf x^*)\Xi\|_1<1.
	\)
\end{lemma}

\begin{proof}
	Let $J^*:=J(\mathbf x^*)	= \kappa^*\bigl(\operatorname{diag}(\mathbf a^*)-\mathbf x^*(\mathbf a^*)^\top\bigr)$ from Lemma~\ref{lem:jacobian-compact}, where $a_j^*:=r_j/(1-x_j^*)^2$.
	For a diagonal matrix $\Xi=\operatorname{diag}(\xi_1,\dots,\xi_n)\succ0$, the $j$-th absolute column sum of $\Xi^{-1}J^*\Xi$ is
	\begin{equation}\label{eq:colsumD}
		\Sigma_j(\Xi) := \sum_i |(\Xi^{-1}J^*\Xi)_{ij}| = \kappa^* a_j^* \Bigl((1-x_j^*)+\sum_{i\ne j}x_i^*\frac{\xi_j}{\xi_i}\Bigr).
	\end{equation}
	
	We construct $\Xi$ by setting $\xi_j=1-\varepsilon$ for $j\in\mathcal M$ and $\xi_i=1$ for $i\notin\mathcal M$, with $\varepsilon\in(0,1)$.
	For any $j\in\mathcal M$, we have $x_j^*=\tfrac12$ and, using the equilibrium relation $x_j^*=\kappa^* r_j/(1-x_j^*)$, one gets $\kappa^* a_j^* = \kappa^* r_j/(1-x_j^*)^2 = x_j^*/(1-x_j^*)=1$.
	Substituting into \eqref{eq:colsumD} gives
	\[
	\Sigma_j(\Xi) = \tfrac12+(1-\varepsilon)\sum_{i\ne j}\frac{x_i^*}{\xi_i}.
	\]
	Since $\xi_i\ge 1-\varepsilon$ and $\sum_{i\ne j}x_i^*=\tfrac12$, one always has $\Sigma_j(\Xi)\le 1$.
	Moreover, because $\mathcal M\neq\mathcal I$ (otherwise $x_i^*=\tfrac12$ for all $i$, which together with $\sum_i x_i^*=1$ would force $n=2$), there exists at least one index $i\notin\mathcal M$ with $\xi_i=1$.
	Hence $\sum_{i\ne j}x_i^*/\xi_i<\tfrac12/(1-\varepsilon)$, 
	and therefore $\Sigma_j(\Xi)<1$ for all $j\in\mathcal M$.
	
	Next, for any $k\notin\mathcal M$, one has $x_k^*<\tfrac12$, hence there exists $\eta\in(0,1]$ such that $2x_k^*\le 1-\eta$.
	Using \eqref{eq:colsumD} and the fact that $\xi_k/\xi_i\le (1-\varepsilon)^{-1}$ whenever $i\in\mathcal M$, we obtain
	\begin{align*}
		\Sigma_k(\Xi)
		&=
		\kappa^* a_k^*
		\Bigl((1-x_k^*)+\sum_{i\ne k}x_i^*\frac{\xi_k}{\xi_i}\Bigr)\\
		&\le
		\kappa^* a_k^*
		\Bigl((1-x_k^*)+\sum_{i\ne k}x_i^*\Bigr)
		+
		\kappa^* a_k^*
		\sum_{i\in\mathcal M}x_i^*
		\Bigl(\tfrac{1}{1-\varepsilon}-1\Bigr)\\
		&=
		2x_k^*+K_{\max}\frac{\varepsilon}{1-\varepsilon},
	\end{align*}
	where $K_{\max}:= \max_k(\kappa^* a_k^*) \sum_{i\in\mathcal M}x_i^*<\infty$.
	Choosing $\varepsilon>0$ sufficiently small so that $K_{\max}\frac{\varepsilon}{1-\varepsilon}<\eta$ ensures $\Sigma_k(\Xi)<1$ for all $k\notin\mathcal M$.
	Therefore $\|\Xi^{-1}J^*\Xi\|_1	= \max_j\Sigma_j(\Xi)<1$.
\end{proof}

\begin{lemma}[Diagonal $\ell_1$-contraction at the mixed-branch equilibrium]\label{lem:diag-weight-mixed}
	Let $\mathbf x^*\in\mathrm{int}\Delta_n$ be the unique mixed-branch equilibrium, i.e., $x_{i_{\max}}^*>\tfrac12$ for a unique dominant agent $i_{\max}$ and $x_j^*<\tfrac12\,\forall j\neq i_{\max}$.
	Then there exists a diagonal matrix $\Xi=\mathrm{diag}(\xi_1,\dots,\xi_n)\succ0$ such that $\|\Xi^{-1}J(\mathbf x^*)\Xi\|_1<1$.
\end{lemma}

\begin{proof}
	Let $J^*:=J(\mathbf x^*)$ and use the column-sum formula \eqref{eq:colsumD}.
	Construct $\Xi(\varepsilon)=\mathrm{diag}(\xi_1,\dots,\xi_n)$ by setting $\xi_{i_{\max}}=1-\varepsilon$ and $\xi_j=1$ for $j\neq i_{\max}$, where $\varepsilon\in(0,1)$.
	A direct calculation gives
	\[
	\Sigma_j(\Xi)=
	\begin{cases}
		2x_{i_{\max}}^*-\varepsilon x_{i_{\max}}^* & \text{if } j=i_{\max},\\
		2x_j^*+\dfrac{x_j^*x_{i_{\max}}^*}{1-x_j^*} \dfrac{\varepsilon}{1-\varepsilon} & \text{if } j\neq i_{\max}.
	\end{cases}
	\]
	
	For the column $j=i_{\max}$, the condition $\Sigma_{i_{\max}}(\Xi)<1$ is equivalent to
	$\varepsilon>\varepsilon_{\min}:=2-1/x_{i_{\max}}^* \in(0,1)$.
	
	For each $j\neq i_{\max}$, define $\delta_j:=1-2x_j^*>0$ and $\nu_j:=x_j^*x_{i_{\max}}^*/(1-x_j^*)>0$. 
	Then $\Sigma_j(\Xi)<1$ is equivalent to
	\[
	2x_j^*+\nu_j\frac{\varepsilon}{1-\varepsilon}<1
	\ \Longleftrightarrow \
	\varepsilon<\varepsilon_j:=\frac{\delta_j}{\delta_j+\nu_j}.
	\]
	Let $\varepsilon_{\max}:=\min_{j\neq i_{\max}}\varepsilon_j\in(0,1]$.
	We show that $\varepsilon_{\max}>\varepsilon_{\min}$ for $n\ge 3$, hence one can always select any $\varepsilon\in(\varepsilon_{\min},\varepsilon_{\max})$ and obtain $\max_j \Sigma_j(\Xi)<1$.

	To this end, fix any $j\neq i_{\max}$.
	Since $\mathbf x^*\in\mathrm{int}\Delta_n$ and $n\ge 3$, there exists at least one index
	$h\notin\{i_{\max},j\}$ such that $x_h^*>0$.
	Hence $x_j^* < \sum_{k\neq i_{\max}}x_k^* =	1-x_{i_{\max}}^*$.
	Therefore, $
	\delta_j=1-2x_j^*
	>
	1-2(1-x_{i_{\max}}^*)
	=
	2x_{i_{\max}}^*-1$,
	and also, since $1-x_j^*>x_{i_{\max}}^*$ and $\nu_j={x_j^*x_{i_{\max}}^*}/{(1-x_j^*)}
	<
	x_j^*
	<
	1-x_{i_{\max}}^*$.
	Because the map $(\delta,\nu)\mapsto \delta/(\delta+\nu)$ is increasing in $\delta$ and decreasing in $\nu$, it follows that
	\begin{equation*}
		\begin{aligned}
			\varepsilon_j=\frac{\delta_j}{\delta_j+\nu_j}
			&>
			\frac{2x_{i_{\max}}^*-1}{(2x_{i_{\max}}^*-1)+(1-x_{i_{\max}}^*)} \\
			&=
			\frac{2x_{i_{\max}}^*-1}{x_{i_{\max}}^*}
			=
			2-\frac{1}{x_{i_{\max}}^*}
			=
			\varepsilon_{\min}.
		\end{aligned}
	\end{equation*}
	Since this holds for every $j\neq i_{\max}$, we conclude that $\varepsilon_{\max}:=\min_{j\neq i_{\max}}\varepsilon_j>\varepsilon_{\min}$.
	Choosing any $\varepsilon\in(\varepsilon_{\min},\varepsilon_{\max})$ yields $\Sigma_j(\Xi)<1$ for all $j$, and therefore $\|\Xi^{-1}J^*\Xi\|_1<1$.
\end{proof}

\begin{lemma}[Structural bound for eigenvector differentials]\label{lem:struct-Dmu}
	Let $\mathcal P(\mathbf x)$ and $\bar{\mathcal P}(\mathbf x)$ be the Riesz projectors at the simple eigenvalue $1$ of $\mathcal A(\mathbf x)$ and $\bar{\mathcal A}(\mathbf x)$, respectively, and let $\bm\mu(\mathbf x)$ and $\bar{\bm\mu}(\mathbf x)$ be the corresponding normalized left eigenvectors.
	Define $E_{\mathcal P}:=\mathcal P(\mathbf x)-\bar{\mathcal P}(\mathbf x)$, $E_{\mu}:=\bm\mu(\mathbf x)-\bar{\bm\mu}(\mathbf x)$,
	and, for any direction $\mathbf h\in\mathbb R^n$, $E_{D\mathcal P}[\mathbf h]
	:=
	D\mathcal P(\mathbf x)[\mathbf h]-D\bar{\mathcal P}(\mathbf x)[\mathbf h]$.
	Then
	\begin{align*}
		\|D\bm\mu(\mathbf x)[\mathbf h]-D\bar{\bm\mu}(\mathbf x)[\mathbf h]\|_1
		\le~& 
		2\|E_{D\mathcal P}[\mathbf h]\|_1 \\
		 +~& K_{\rm gap}\|\mathbf h\|_1
		\bigl(\|E_{\mathcal P}\|_1+\|E_{\mu}\|_1\bigr),
	\end{align*}
	where $K_{\rm gap}:=2M_{\mathcal P}+2M_{\bar{\mathcal P}}+\bar\tau_*$,
	with $M_{\mathcal P}:=\sup_{\|\mathbf h\|_1=1}\|D\mathcal P(\mathbf x)[\mathbf h]\|_1$, $M_{\bar{\mathcal P}}:=\sup_{\|\mathbf h\|_1=1}\|D\bar{\mathcal P}(\mathbf x)[\mathbf h]\|_1$, and $\bar\tau_*:=
	\sup_{\|\mathbf h\|_1=1}
	\left|
	\bar{\bm\mu}(\mathbf x)^\top
	D\bar{\mathcal P}(\mathbf x)[\mathbf h]
	(I-\bar{\mathcal P}(\mathbf x))\mathds{1}
	\right|$.
\end{lemma}

\begin{proof}
	For brevity, all quantities below are evaluated at the same $\mathbf x$.
	From $\bm\mu^\top\mathcal P=\bm\mu^\top$, differentiation along $\mathbf h$ gives
	\[
	D\bm\mu^\top[\mathbf h](I-\mathcal P)=\bm\mu^\top D\mathcal P[\mathbf h](I-\mathcal P).
	\]
	Hence, $D\bm\mu^\top[\mathbf h]=U+\tau(\mathbf h)\bm\mu^\top$, with $U:=\bm\mu^\top D\mathcal P[\mathbf h](I-\mathcal P)$.
	Using the constraint $D\bm\mu^\top[\mathbf h]\mathds{1}=0$ (owing to $\bm\mu^\top\mathds{1}=1$), one has $\tau(\mathbf h)=-U\mathds{1}$, and thus $D\bm\mu^\top[\mathbf h] = U-\bigl(U\mathds{1}\bigr)\bm\mu^\top$.
	Analogously, $D\bar{\bm\mu}^\top[\mathbf h] = \bar U-\bar\tau(\mathbf h)\bar{\bm\mu}^\top$, with $\bar U:=\bar{\bm\mu}^\top D\bar{\mathcal P}[\mathbf h](I-\bar{\mathcal P})$ and $\bar\tau(\mathbf h):=\bar U\mathds{1}$.
	Therefore
	\[
	D\bm\mu^\top[\mathbf h]-D\bar{\bm\mu}^\top[\mathbf h]
	=
	(U-\bar U)-\bigl(\tau(\mathbf h)\bm\mu^\top-\bar\tau(\mathbf h)\bar{\bm\mu}^\top\bigr).
	\]
	First, expanding $U-\bar U$ yields $U-\bar U
	=
	E_{\mu}^\top D\bar{\mathcal P}[\mathbf h](I-\bar{\mathcal P})
	+
	\bm\mu^\top E_{D\mathcal P}[\mathbf h](I-\bar{\mathcal P})
	+
	\bm\mu^\top D\mathcal P[\mathbf h](\bar{\mathcal P}-\mathcal P)$.
	Hence,
	\begin{equation}\label{eq:UminusUbar-new}
		\begin{aligned}
			\|U-\bar U\|_1
			\le 
			\|E_{\mu}\|_1 M_{\bar{\mathcal P}} \|\mathbf h\|_1
			+
			\|E_{D\mathcal P}[\mathbf h]\|_1 +
			\|E_{\mathcal P}\|_1 M_{\mathcal P} \|\mathbf h\|_1.
		\end{aligned}
	\end{equation}
	Next, one can write $\tau(\mathbf h)\bm\mu^\top-\bar\tau(\mathbf h)\bar{\bm\mu}^\top = (\tau(\mathbf h)-\bar\tau(\mathbf h))\bm\mu^\top
	+\bar\tau(\mathbf h)(\bm\mu^\top-\bar{\bm\mu}^\top)$.
	Therefore,
	\[
	\|\tau(\mathbf h)\bm\mu^\top-\bar\tau(\mathbf h)\bar{\bm\mu}^\top\|_1
	\le
	|\tau(\mathbf h)-\bar\tau(\mathbf h)|
	+
	|\bar\tau(\mathbf h)|\,\|E_{\mu}\|_1.
	\]
	By the same expansion leading to \eqref{eq:UminusUbar-new}, one derives
	\begin{align*}
		|\tau(\mathbf h)-\bar\tau(\mathbf h)|
		\le
		\|E_{D\mathcal P}[\mathbf h]\|_1 +
		\bigl(
		M_{\bar{\mathcal P}} \|E_{\mu}\|_1
		+
		M_{\mathcal P} \|E_{\mathcal P}\|_1
		\bigr)\|\mathbf h\|_1.
	\end{align*}
	Moreover, by definition of $\bar\tau_*$, one obtains $|\bar\tau(\mathbf h)|\le \bar\tau_*\|\mathbf h\|_1$.
	Combining the above estimates yields the claimed bound.
\end{proof}

\begin{lemma}[Local Lipschitz stability]\label{lem:lip-blackbox}
	Let Assumptions~\ref{assum:robust-ref}--\ref{assum:robust-small} hold.
	Fix any compact set $\mathcal N\subset\mathrm{int}\Delta_n$ and any norm $\|\cdot\|_{\star}$ on $\mathbb R^n$.
	Then there exists a constant $K_{\rm lip}^{\star}>0$, depending only on $\delta_A$, $R$, $\bar{\mathcal C}_{[n]}$, the set $\mathcal N$, and the chosen norm $\|\cdot\|_{\star}$, such that
	\[
	\mathrm{Lip}_{\mathcal N}^{\star}(\bm\phi)
	\le
	\mathrm{Lip}_{\mathcal N}^{\star}(\bm\alpha)+K_{\rm lip}^{\star}\widetilde\varepsilon,
	\]
	where $\widetilde\varepsilon:=\varepsilon+\varepsilon_C\|I-R\|_1$ and $\varepsilon_C:=\max_i\|C_i-\bar C_i\|_1$.
	In particular, if $\bm\alpha$ is $q$-contractive on $\mathcal N$ in $\|\cdot\|_{\star}$ and $\widetilde\varepsilon$ is small enough that $q+K_{\rm lip}^{\star}\widetilde\varepsilon<1$,
	then $\bm\phi$ is also a contraction on $\mathcal N$ in the norm $\|\cdot\|_{\star}$.
\end{lemma}

\begin{proof}
	The proof strategy is to control $D\bm\phi-D\bm\alpha$ on $\mathcal N$ by bounding the difference between the Fr\'echet derivatives of the associated Riesz projectors.

	Let $R_A(z):=(zI-\mathcal A(\mathbf x))^{-1}$ and $R_{\bar A}(z):=(zI-\bar{\mathcal A}(\mathbf x))^{-1}$.
	For a simple isolated eigenvalue, the Fr\'echet derivative of the Riesz projector admits the contour representation
	\cite{kato2013perturbation}:
	\begin{equation}\label{eq:DP-resolvent}
		D\mathcal P(\mathbf x)[\mathbf h]
		=
		\frac{1}{2\pi i}\oint_{\Gamma}R_A(z)\bigl(D\mathcal A(\mathbf x)[\mathbf h]\bigr)R_A(z) \mathrm{d}z.
	\end{equation}
	Moreover, one has $D\mathcal A(\mathbf x)[\mathbf h] = \mathcal C_{[n]}((DW(\mathbf x)[\mathbf h])\otimes I_m)$ and $D\bar{\mathcal A}(\mathbf x)[\mathbf h] = \bar{\mathcal C}_{[n]}((DW(\mathbf x)[\mathbf h])\otimes I_m)$,
	where $DW(\mathbf x)[\mathbf h]=\operatorname{diag}(\mathbf h)(I-R)$.
	Hence, for $\|\mathbf h\|_1=1$, one has $\|D\mathcal A(\mathbf x)[\mathbf h]-D\bar{\mathcal A}(\mathbf x)[\mathbf h]\|_1 \le \varepsilon_C\|I-R\|_1$.
	By \eqref{eq:DP-resolvent} and the identity
	\begin{equation*}
		\begin{aligned}
			R_A D\mathcal A R_A - R_{\bar A} & D\bar{\mathcal A} R_{\bar A}
			= R_A(D\mathcal A-D\bar{\mathcal A})R_A + \\
			&  (R_A-R_{\bar A})D\bar{\mathcal A}R_A 
			+R_{\bar A}D\bar{\mathcal A}(R_A-R_{\bar A}),
		\end{aligned}
	\end{equation*}
	one derives $\|D\mathcal P[\mathbf h]-D\bar{\mathcal P}[\mathbf h]\|_1
	\le
	K_1\varepsilon_C\|I-R\|_1+K_2\varepsilon
	\le
	K_3\widetilde\varepsilon$
	for some constants $K_1,K_2,K_3>0$ (depending only on $\delta_A$, $R$, $\bar{\mathcal C}_{[n]}$, and $\mathcal N$), where we used Lemma~\ref{lem:robust-res}, Assumption~\ref{assum:robust-small}, and the identity $R_A-R_{\bar A}=R_A(\mathcal A-\bar{\mathcal A})R_{\bar A}=R_AER_{\bar A}$.
	
	Applying Lemma~\ref{lem:struct-Dmu}, together with Lemma~\ref{lem:rankone-proj}, this implies $\sup_{\mathbf x\in\mathcal N, \|\mathbf h\|_1=1}
	\|D\bm\mu(\mathbf x)[\mathbf h]-D\bar{\bm\mu}(\mathbf x)[\mathbf h]\|_1
	\le
	K_4\widetilde\varepsilon$
	for some constant $K_4>0$.
	Since, by \eqref{equ:AggregatedPower}, one gets $\bm\phi(\mathbf x)=B\bm\mu(\mathbf x)$ and $\bm\alpha(\mathbf x)=B\bar{\bm\mu}(\mathbf x)$ with $B:=I_n\otimes \mathds{1}_m^\top$ (so $\|B\|_1=1$). Then one derives $\sup_{\mathbf x\in\mathcal N, \|\mathbf h\|_1=1}
	\|D\bm\phi(\mathbf x)[\mathbf h]-D\bm\alpha(\mathbf x)[\mathbf h]\|_1
	\le
	K_4\widetilde\varepsilon$.
	
	Since all norms on $\mathbb R^n$ are equivalent, the induced operator norms are also equivalent.
	For the fixed norm $\|\cdot\|_{\star}$, there is a constant $\nu_{\star}^{\rm op}>0$ such that $\|M\|_{\star}\le \nu_{\star}^{\rm op}\|M\|_1$ for every linear map $M:\mathbb R^n\to\mathbb R^n$.
	Hence $\sup_{\mathbf x\in\mathcal N}\|D\bm\phi(\mathbf x)-D\bm\alpha(\mathbf x)\|_{\star}
	\le
	\nu_{\star}^{\rm op}K_4\widetilde\varepsilon.$
	Then $\mathrm{Lip}_{\mathcal N}^{\star}(\bm\phi)
	\le
	\mathrm{Lip}_{\mathcal N}^{\star}(\bm\alpha)+K_{\rm lip}^{\star}\widetilde\varepsilon$
	with $K_{\rm lip}^{\star}:=\nu_{\star}^{\rm op}K_4$.
	The final contraction statement is immediate.
\end{proof}

\bibliographystyle{IEEEtran}
\bibliography{IEEEabrv,MultiDF_repelling_bibfile}

\end{document}